\documentclass[letterpaper,12pt]{article}
\usepackage[latin1]{inputenc}
\usepackage{amsmath, graphicx, layout, verbatim,rotating}
\usepackage{setspace}
\usepackage{fancyhdr}
\usepackage{float}
\usepackage{color}
\usepackage{multirow}
\usepackage{natbib}
\bibpunct[; ]{(}{)}{;}{a}{,}{;}
\usepackage{verbatim}
\usepackage{subfig}
\usepackage{amsthm}
\usepackage{amsfonts}
\usepackage{algorithm}
\usepackage{algpseudocode}
\usepackage{amssymb}
\usepackage[colorlinks=true,urlcolor=blue,citecolor=blue	,linkcolor=blue]{hyperref}

  \usepackage{enumitem}
\setlist{nolistsep}

\renewcommand{\algorithmicrequire}{\textbf{\small Input:}}
\renewcommand{\algorithmicensure}{\textbf{\small Output:}}

\usepackage{setspace}

\def\X{***ATT***}

\renewcommand{\vec}[1]{\mathbf{#1}}
\renewcommand{\comment}[1]{}

\def\I{{\mathbb I}}

\def\E{{\mathbb E}}
\def\V{{\mathbb V}}
\def\X{{\vec{X}}}
\def\x{{\vec{x}}}
\def\t{{\vec{t}}}
\def\Z{{\vec{Z}}}
\def\z{{\vec{z}}}
\def\y{{\vec{y}}}
\def\Re{{\mathbb R}}

\newtheorem{thm}{Theorem}

\newtheorem{Cor}{Corollary}
\newtheorem{Lemma}{Lemma}
\newtheorem{Assumption}{Assumption}

\usepackage{xpatch}
\xpatchcmd{\algorithmic}{\setcounter}{\algorithmicfont\setcounter}{}{}
\providecommand{\algorithmicfont}{}
\providecommand{\setalgorithmicfont}[1]{\renewcommand{\algorithmicfont}{#1}}

\theoremstyle{definition}

\providecommand{\keywords}[1]{\textbf{\textbf{Key Words:}} #1}

\begin{document}

\setalgorithmicfont{\small}

\title{Converting High-Dimensional Regression to High-Dimensional Conditional Density Estimation}
\date{}
\author{Rafael Izbicki\thanks{Department of Statistics, Federal University of São Carlos, Brazil.} \ and Ann B. Lee\thanks{Department of Statistics, Carnegie Mellon University, USA.}}

\maketitle 

\begin{abstract} 

 There is a growing demand for nonparametric conditional density estimators (CDEs) in fields such as astronomy and economics.  In astronomy, for example, one can dramatically improve estimates of the parameters that dictate the  evolution of the Universe by working with full conditional densities instead of regression (i.e., conditional mean) estimates. More generally, standard regression falls short in any prediction problem where the distribution of the response is more complex with multi-modality, asymmetry or heteroscedastic noise. Nevertheless, much of the work on high-dimensional inference concerns regression and classification only, whereas research on density estimation has lagged behind. Here we propose {\em FlexCode}, a fully nonparametric approach to conditional density estimation that reformulates CDE as a non-parametric orthogonal series problem where the expansion coefficients are estimated by regression. By taking such an approach, one can efficiently estimate conditional densities and not just expectations in high dimensions by drawing upon the success in high-dimensional regression. Depending on the choice of regression procedure, our method can adapt to a variety of challenging high-dimensional settings with different structures in the data (e.g., a large number of irrelevant components and nonlinear manifold structure) as well as different data types (e.g., functional data, mixed data types and sample sets). We study the theoretical and empirical performance of our proposed method, and we compare our approach with traditional conditional density estimators on simulated as well as real-world data, such as photometric galaxy data, Twitter data, and  line-of-sight velocities in a galaxy cluster.
 

\end{abstract}

  \section{Introduction}
\label{intro}


A challenging problem in modern statistical inference is how to estimate
a conditional density of a random variable $Z \in \Re$ 
given a high-dimensional random vector $\vec{X} \in \Re^D$, $f(z|\x)$. 
This quantity plays a key role in several statistical problems in the sciences where
the regression function $\E[Z|\x]$ is not informative enough due to multi-modality and asymmetry of the conditional density. 

For example, several recent works in 
cosmology \citep{Sheldon,kind2013tpz,Rau01102015} have shown that one can significantly reduce systematic errors in cosmological analyses by using the full probability distribution of photometric redshifts $z$ (a key quantity that relates the distance of a galaxy to the observer) 
given galaxy colors $\vec{x}$ (i.e., differences of brightness measures at two different wavelengths). Other fields where conditional density estimation plays a key role are time series forecasting in economics \citep{kalda2013nonparametric} and   approximate Bayesian methods
\citep{fan2013approximate,izbickiLeeSchafer,papamakarios2016fast}.
 Conditional densities can also be used to 
construct accurate predictive intervals for new observations in settings with complicated sources of errors \citep{Fernandez-Soto} or multimodal distributions (see Fig.~\ref{fig::Asymmetric} and Fig.~\ref{fig::clusterCoverage} for examples).

\begin{figure}[H]
	\centering
	\subfloat{  \includegraphics[page=2,scale=0.29]{predictionBaunds1DBimodal.pdf}} 
	\subfloat{  \includegraphics[page=1,scale=0.29]{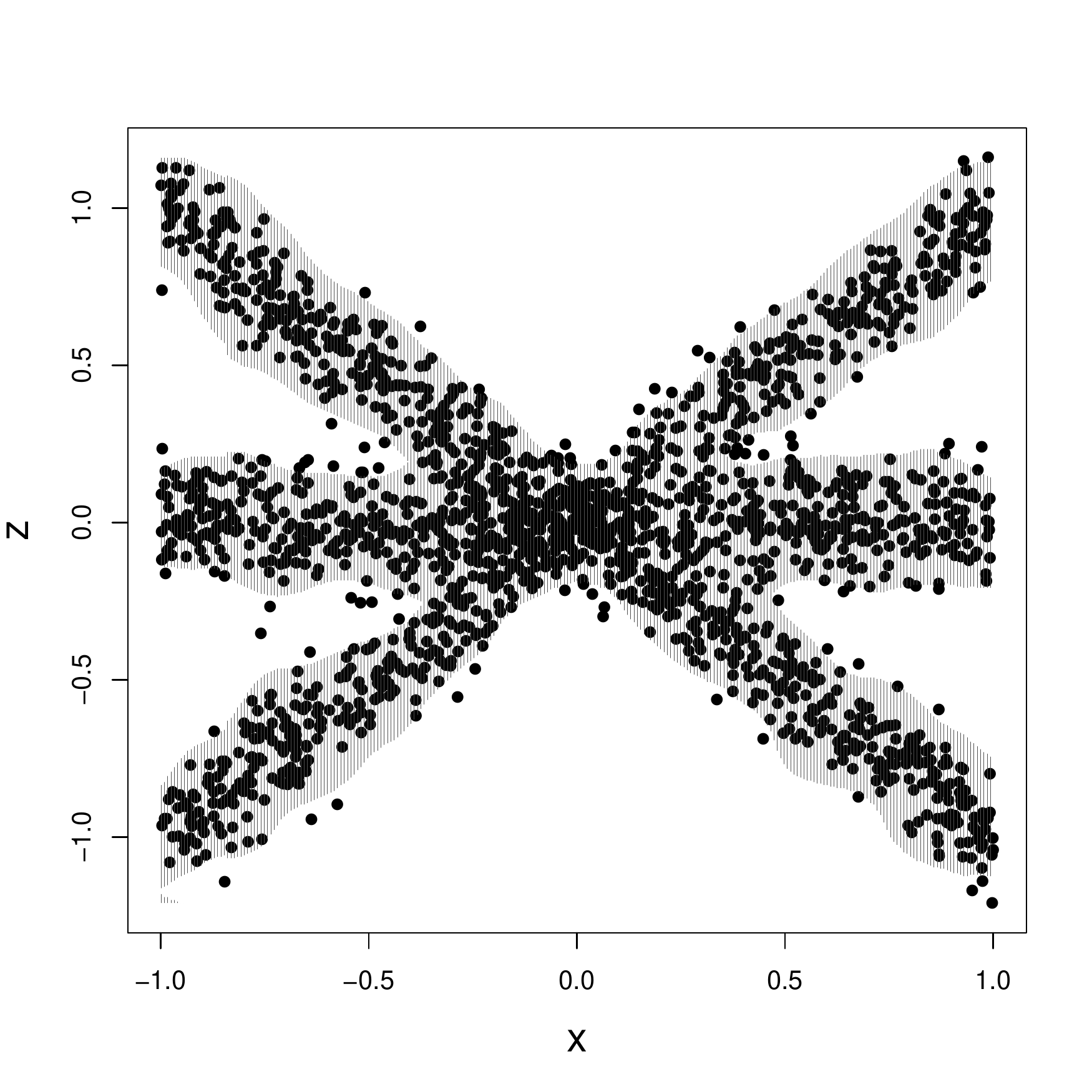}} 
	\vspace{-3mm}
	\caption{\footnotesize  A toy example where nonparametric regression methods fail to capture the underlying structure and return too wide prediction bands, whereas a nonparametric conditional density estimator automatically returns informative predictive bands. The left plot shows  95\%  predictions bands from local linear regression, and the right plot shows 95\% highest predictive density (HPD) bands derived from \emph{FlexCode-SAM} estimates of the conditional density.} 
	\label{fig::Asymmetric}
\end{figure}

Nevertheless,  whereas  a large literature has been devoted to estimating 
the regression $\E[Z|\x]$, statisticians have paid far less attention to estimating the full conditional density $f(z|\x)$, especially when $\x$ is high-dimensional.
Most attempts to estimate $f(z|\x)$ can effectively only handle up to about 3 covariates (see, e.g., \citealt{Fan2}). In higher dimensions, such methods typically rely on a prior dimension reduction step 
 which, as is the case with any data reduction, can result in significant loss of information. 
\vspace{2mm}

\textbf{Contribution.} There is currently no {\em general} procedure for converting successful regression estimators (that is, estimators of the conditional mean $\E[Z|\x]$) to estimators of the full conditional density $f(z|\x)$ --- indeed, this is a non-trivial problem. In this paper, we propose a fully nonparametric approach to conditional density estimation, which reformulates CDE as an orthogonal series problem where the expansion coefficients are estimated by regression. By taking such an approach, one can efficiently estimate conditional densities in high dimensions by drawing upon the success in high-dimensional regression. Depending on the choice of regression procedure, our method can exploit different types of sparse structure in the data, 
as well as handle different types of data.

For example, in a setting with submanifold structure, our estimator adapts to the intrinsic dimensionality of the data with a suitably chosen regression method;  such as, nearest neighbors, local linear, tree-based or spectral series regression \citep{Bickel:Li:2007,kpotufe2011k,kpotufe2012tree,LeeIzbickiReg}.
Similarly, if 
the number of relevant covariates (i.e., covariates that  affect the distribution of $Z$) is small, one can construct a good conditional density estimator using 
 lasso, SAM, Rodeo or other additive-based regression estimators \citep{tibshirani1996regression,lafferty2008rodeo,meier2009high,yang2015minimax}. 
Because of the flexibility of our approach, the method is able to overcome the  the curse of dimensionality in a variety of scenarios with faster convergence rates and better performance than traditional conditional
density estimators; see Sections~\ref{sec::appli}-\ref{sec::theory} for specific examples and analysis. 
By choosing appropriate regression methods, the method can also handle
different types of covariates that represent discrete data,
mixed data types, functional data, circular data, and so on, which generally require
hand-tailored 
techniques (e.g., \citealt{di2016note}).  Most notably, Sec.~\ref{sec::distributional} describes an entirely new area of conditional density estimation (here referred to as ``Distribution CDE'') where a predictor is an entire {\em sample set} from an underlying distribution.

 We call our general approach {\em FlexCode}, which stands for {\em Flex}ible nonparametric {\em co}nditional {\em d}ensity {\em e}stimation via regression.
\vspace{2mm}

{\bf Existing Methodology.}  With regards to existing methods for estimating $f(z|\vec{x})$, several nonparametric estimators have been proposed   
 when $\vec{x}$  lies in  a {\em low-dimensional} space.
 Many of these methods are based on first estimating $f(z,\vec{x})$ and $f(\vec{x})$ with, for example, kernel density estimators \citep{Rosenblatt}, and then combining the estimates according to $f(z|\vec{x})=\frac{f(z,\vec{x})}{f(\vec{x})}$.
Several works further improve upon such an approach by using different criteria and shortcuts to tune parameters as well as creating fast shortcuts to 
implement these methods (e.g.,
\citealt{Hyndman,Ichimura}).
Other approaches to conditional density estimation in low dimensions
include using locally polynomial regression \citep{FanYaoTon96},  least squares approaches \citep{sugiyama2010conditional} and density estimation through quantile estimation \citep{Takeuchi2}; see \citet{bertin2016adaptive}
and references therein for other methods. 

For moderate dimensions, \citet{Hall2} propose a method for tuning parameters in kernel density 
estimators which automatically determines which components of $\vec{x}$ are 
relevant to $f(z|\vec{x})$. The method produces good results  
but is not practical for high-dimensional data sets: Because it relies on choosing a different bandwidth for each covariate, it has 
a high computational cost that increases with both the sample size $n$ and the dimension $D$, with prohibitive costs even for moderate $n$'s and $D$'s. 
Similarly, 
\citet{shiga2015direct} propose a conditional estimator
that selects relevant components but under the restrictive assumption that $f(z|\x)$ has an additive structure; moreover
the method scales as $O(D^3)$, 
which is also computationally prohibitive for moderate dimensions.
Another framework is developed by
\citet{Efromovich3}, who proposes an orthogonal series estimator that automatically performs dimension reduction on $\vec{x}$
when several components of this vector are conditionally independent of the response.
Unfortunately, the method requires one to compute $D+1$ tensor products, which quickly becomes
computationally intractable even for as few as 10 covariates. 
More recently, \citet{IzbickiLeeCDE} propose an alternative orthogonal series estimator that uses a basis that adapts to the geometry of the data. They show that their approach, called Spectral Series CDE,  as well as the k-nearest neighbor method by \citet{Lincheng},  work
well in high dimensions when there is submanifold structure. These methods, however,
do not perform well when $\x$ has irrelevant components.

FlexCode, on the other hand, is flexible enough to overcome the difficulties of other methods under
a large variety of situations because it makes use of the many existing regression methods for high-dimensional inference. As an example, Fig.~\ref{fig::twitter} shows the level sets of the estimated conditional
density in a 
challenging problem that involves $\approx$500 covariates.
Here we estimate $f(\z|\x)$,
where $\x$ is the content of a tweet and
$\z$ is the location where it was posted (latitude and longitude).\comment{
While FlexCode is able to estimate
the location of tweets even in ambiguous cases
(there is a Long Beach both in California and in Connecticut,
which is reflected by the results in the 
bottom right plot in Fig.~\ref{fig::twitter}),
no other existing fully nonparametric  methods we are aware of are able
to estimate this quantity with reasonable precision.}
FlexCode, based on sparse additive regression, is able to estimate
the location of tweets even in ambiguous cases
(there is a Long Beach both in California and in Connecticut,
which is reflected by the results in the 
bottom right plot in Fig.~\ref{fig::twitter}); we are not aware of any other 
 existing fully nonparametric  method that are able
to estimate this quantity with reasonable precision as well as attach meaningful measures of uncertainty. See more details about this example in Sec.~\ref{sec::twitter}.

\begin{figure}[H]
	\centering
	\subfloat{  \includegraphics[page=1,scale=0.41]{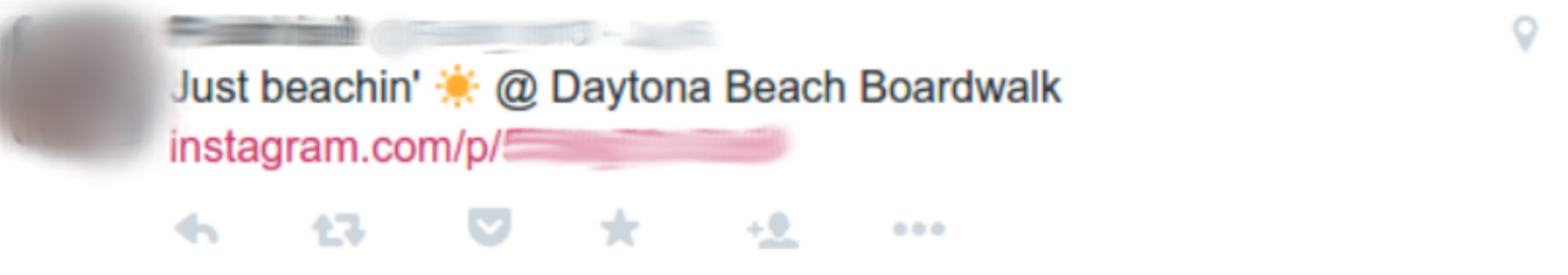}}
	\subfloat{  \includegraphics[page=1,scale=0.41]{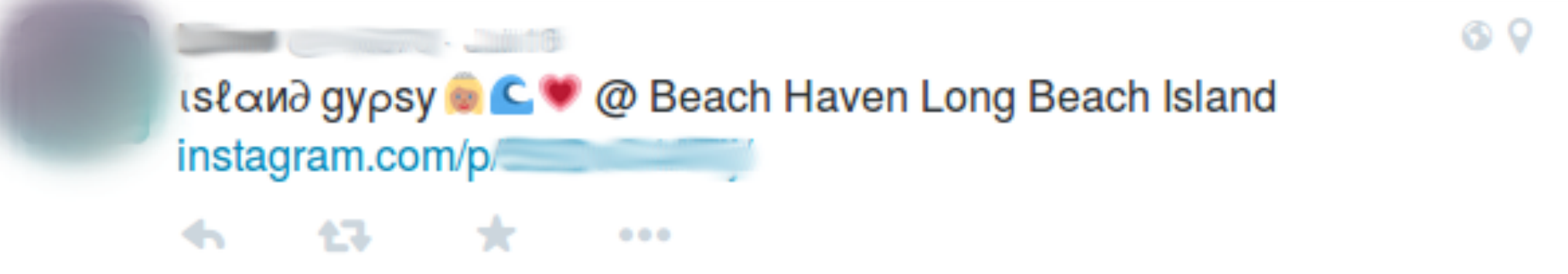}} \\
	\subfloat{  \includegraphics[page=3,scale=0.4]{twitterDataBeachReduced.pdf}}\hspace{10mm} 
	\subfloat{  \includegraphics[page=1,scale=0.4]{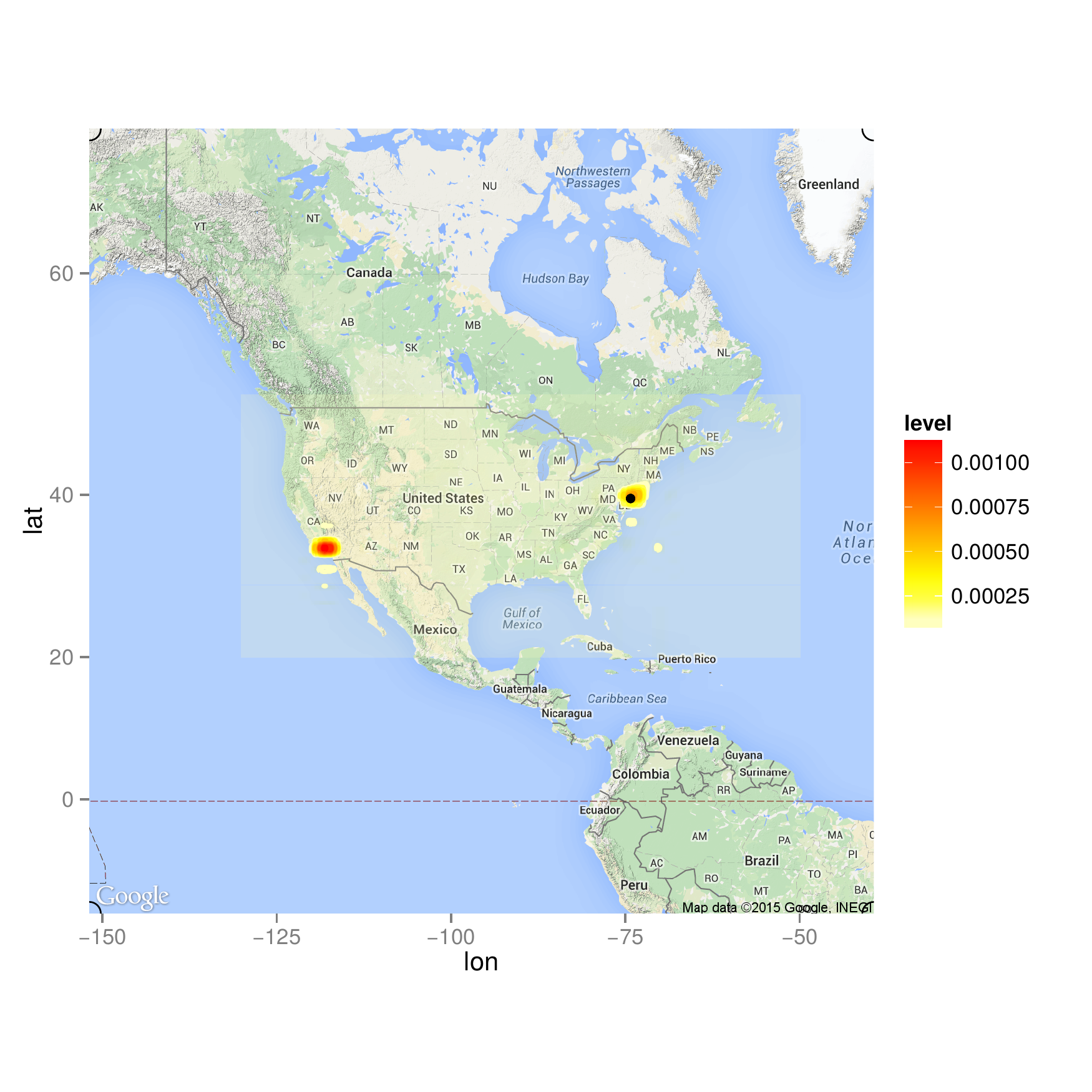}} 
	\vspace{-3mm}
	\caption{\footnotesize  Top: Two tweets with the keyword ``beach''. Bottom: Level sets 
	of the estimated probability density of the tweet locations given the content of the tweets. The black dots indicate their true locations. See Sec.~\ref{sec::twitter} for details.}
	\label{fig::twitter}
\end{figure}

In Section \ref{sec::method}, we describe our method in detail, and  present connections with existing 
literature on Varying Coefficient methods and  Spectral Series CDE. Section \ref{sec::appli} presents several applications of FlexCode. Section \ref{sec::theory} 
discusses convergence rates of the estimator, and Section \ref{sec::concl} concludes the paper.
 
\section{Methods} 
 \label{sec::method}

Assume we observe i.i.d.~data $(\X_1,Z_1),\ldots,(\X_n,Z_n)$, where the covariates $\x \in \mathbb{R}^D$ with $D$ potentially large, and the response $Z \in [0,1]$.~\footnote{More generally, $\x$ can represent functional data, distributions, as well as mixed continuous and discrete data; see Sec.~\ref{sec::appli} for examples. The response $z$ can also be multivariate (Sec.~\ref{sec::extension}) or discrete \citep[Sec.~4.2]{IzbickiLeeCDE}.} Our goal is to estimate the full density $f(z|\x)$ rather than, e.g., only the conditional mean $\E[Z|\x]$ and conditional variance $\V[Z|\x]$. We propose a novel ``varying coefficient'' series approach, where we start by 
 specifying an orthonormal basis $(\phi_i)_{i \in \mathbb{N}}$  for $\mathcal{L}^2(\Re)$.
 This basis will be used to model the density $f(z|\x)$ {\em as a function of $z$}.   As we shall see, each coefficient in the expansion can be directly estimated via a regression. 
  Note that there is a wide range of (orthogonal) bases one can choose from to capture any challenging shape of the density function of interest 
  \citep{mallat1999wavelet}.
For instance, a natural choice for reasonably smooth functions $f(z|\x)$ is the Fourier basis:
\[
    \phi_1(z)=1;\hspace{9mm} \phi_{2i+1}(z)=\sqrt{2}\sin{\left(2\pi iz \right)}, \ i\in \mathbb{N}; \hspace{9mm} \phi_{2i}(z)=\sqrt{2}\cos{\left(2\pi iz \right)}, \ i\in \mathbb{N}
\] 
Alternatively, one can use wavelets or related bases to capture inhomogeneities in the density (see Sec.\ref{sec::distributional} for an example), and indicator functions to model discrete responses~\citep[Sec.~4.2]{IzbickiLeeCDE}.

Smoothing using orthogonal functions is per se not a new concept \citep{efromovich,wasserman}. The novelty in FlexCode is that we, by using an orthogonal series approach for the response variable, can convert a challenging high-dimensional conditional density estimation problem to a simpler high-dimensional regression (point estimation) problem.

For fixed $\x \in \Re^D$,
we write
\begin{equation}
\label{eq::FlexCode}
f(z|\x)=\sum_{i \in \mathbb{N}}\beta_{i }(\x)\phi_i(z).
\end{equation}
Note that our model is fully nonparametric: Equation \ref{eq::FlexCode} hold
as long as, for every $\x$, $f(z|\x)$ is $\mathcal{L}^2(\Re)$ integrable as a function of $z$. Furthermore, 
 because the $\{\phi_i\}_{i \in \mathbb{N}}$  basis functions are orthogonal to each other, the expansion coefficients are given by
\begin{align}
\label{eq::coeff}
\beta_{i }(\x) =\langle f(.|\x),\phi_i\rangle = \int_{\Re} \phi_i(z) f(z|\x) dz = \E\left[\phi_i(Z)|\x\right].
\end{align}
That is, each ``varying coefficient'' $\beta_{i }(\x)$ in Eq.~\ref{eq::FlexCode} is a {\em regression function}, or conditional expectation. 
This suggests that we, for fixed $i$, estimate $\beta_{i }(\x)$ by regressing
$\phi_i(z)$ on $\x$ using the sample $(\X_1,\phi_i(Z_1)),\ldots,(\X_n,\phi_i(Z_n))$. 

We define our FlexCode estimator of $f(z|\x)$  as
\begin{align}
\label{eq:cdeEst}
\widehat{f}(z|\x)=\sum_{i=1}^I\widehat{\beta}_{i}(\x)\phi_i(z),
\end{align}
where the results from the regression, $$\widehat{\beta}_{i }(\x) = \widehat{\E}\left[\phi_i(Z)|\x\right],$$ model how the density varies in covariate space. The cutoff $I$ in the series expansion is a tuning parameter that controls the bias-variance tradeoff in the final density estimate. Generally speaking, the smoother the density, the smaller the value of $I$; see Sec.~\ref{sec::theory} Theory for details. In practice, we use cross-validation or data splitting (Sec.~\ref{lossTuning}) to tune parameters. 

With FlexCode, the problem of high-dimensional conditional density estimation boils down to choosing appropriate methods for estimating the  regression functions $\E \left[\phi_i(Z)|\x\right],$ $i=1,\ldots,I$. The key advantage of FlexCode is its flexibility: By taking advantage of new and existing regression methods, 
 we can adapt to {\em different structures} in the data (e.g., manifolds, irrelevant covariates as well as different relationships between $\x$ and the response $Z$), and we can handle {\em different types of data} (e.g. mixed data, functional data, and so on). 
We will further explore this topic in Secs.~\ref{sec::appli}-\ref{sec::theory}.


\subsection{Loss Function and Tuning of Parameters}
\label{lossTuning}

For a given estimator $\widehat{f}(z|\vec{x})$, we measure the discrepancy between  $\widehat{f}(z|\vec{x})$ and $f(z|\vec{x})$ via the loss function 
\begin{align}
\label{loss} \nonumber 
 L(\widehat{f},f) &= \iint \left(\widehat{f}(z|\vec{x})-f(z|\vec{x})\right)^2dP(\vec{x})dz \\ 
 &=\iint \widehat{f}^2(z|\vec{x})dP(\vec{x})dz-2\iint \widehat{f}(z|\vec{x})f(z,\vec{x}) d\vec{x}dz+C, 
\end{align}
where $C$ is a constant that does not depend on the estimator. 

To tune the parameter $I$, we split the data into a training and a validation set.  
We use the training set to estimate each regression function $\beta_{i}(\x)$. We then use the validation set $(z'_1,\vec{x}'_1),\ldots, (z'_{n'},\vec{x}'_{n'})$ to 
estimate the loss (\ref{loss}) (up to the constant $C$) according to:
\begin{align}
\label{lossEmpirical}
\widehat{L}(\widehat{f},f)=\sum_{i=1}^I \frac{1}{n'} \sum_{k=1}^{n'}\widehat{\beta}^2_{i}(\x_k') -2\frac{1}{n'}\sum_{k=1}^{n'} \widehat{f}(z'_k|\vec{x}'_k),
\end{align}
This estimator is consistent because of the orthogonality of the basis 
$\{\phi_i\}_i$.
We choose the tuning parameters 
with the smallest estimated loss $\widehat{L}(\widehat{f},f)$.
 Algorithm 1 summarizes our procedure.  In line 3, 
we split the training data in two parts to tune the parameters associated with the regression using the standard $\mathcal{L}^2(\Re)$ regression loss, i.e., $\E[(W-\widehat{\beta}_{i }(\X))^2]$.
 
   \begin{algorithm}
  \caption{ \small FlexCode }\label{algorithmCondReg}
  \algorithmicrequire \ {\small Training data; validation data; maximum cutoff $I_0$; orthonormal basis $\{\phi_i\}_i$; regression method and grid of tuning parameters for regression. } 
  
  \algorithmicensure \ {\small Estimator $\widehat{f}(z|\vec{x})$}
  \begin{algorithmic}[1]
     \ForAll{$i \leq I_0$}
     \State Compute $\mathcal{D}=(\X_1,W_1),\ldots,(\X_n,W_n)$, where $W_k:=\phi_i(Z_k)$ 
        \State Estimate the regression $\beta_{i }(\x)=\E\left[W|\x\right]$ using $\mathcal{D}$. 
     \EndFor
        \ForAll{$I \leq I_0$}
  	\State Calculate the estimated loss $\widehat{L}(\widehat{f}_{I},f)$ on the validation set \Comment{Eq.~(\ref{lossEmpirical})} 
  	\State \Comment{$\widehat{f}_{I}$ is the estimator in Eq.~(\ref{eq:cdeEst})}
        \EndFor
      \State Define $\widehat{f}(z|\x)=\arg \min_{\widehat{f}_{I}} \widehat{L}(\widehat{f}_{I},f)$
     \State \textbf{return} $\widehat{f}(z|\vec{x})$
  
  \end{algorithmic}
  \end{algorithm}


 In terms of computational efficiency, FlexCode is typically faster than existing methods for conditional density estimation (see Section~\ref{sec::appli}), especially in high dimensions and for massive data sets. If the FlexCode estimator is based on a scalable regression procedure  (e.g., \citealt{raykar2007scalable,desai2010gear,zhang2013divide,dai2014scalable}), then the resulting conditional density estimator will scale as well. Furthermore, FlexCode is naturally suited for parallel computing, as one can estimate each of the $I$ regression functions 
   separately and then combine the estimates according to Eq.~(\ref{eq:cdeEst}). Our implementation of FlexCode
   is available at \url{https://github.com/rizbicki/FlexCoDE}, and implements a parallel version of the estimator. For the final density estimate (Step 9 in Algorithm 1), we apply the same techniques as in \citet[Section 2.2]{IzbickiLeeCDE}   to remove potentially negative values and spurious bumps.

  
 \subsection{Extension to Vector-Valued Responses}
 \label{sec::extension}

 By tensor products, one can directly extend FlexCode to cases where the response variable $\Z$ is vector-valued.
 For instance,
 if $\Z \in \Re^2$, consider the basis
 $$\left\{\phi_{i,j}(\vec{z})=\phi_{i}(z_1)\phi_{j}(z_2): i,j\in \mathbb{N}\right\},$$
 where $\z=(z_1,z_2)$, and $\left\{\phi_{i}(z_1)\right\}_{i}$ and $\left\{\phi_{j}(z_2)\right\}_{j}$ are bases for functions in $\mathfrak{L}^2(\Re)$.  Then, let
  $$f(\z|\x)=\sum_{i,j \in \mathbb{N}}\beta_{i,j}(\x)\phi_{i,j}(\z),$$
 where the expansion coefficients
 \begin{align*}
 \beta_{i,j}(\x) =\langle f(.|\x),\phi_{i,j}\rangle = \int_{\Re^2} \phi_{i,j}(\z) f(\z|\x) d\z = \E\left[\phi_{i,j}(\Z)|\x\right].
 \end{align*}
 Note that each $\beta_{i }(\x)$ is a regression function of a \emph{scalar} response. In other words,  the FlexCode framework
  allows one to estimate multivariate conditional densities by only using regression estimators of scalar responses.

 
 \emph{Remark:} To avoid tensor products, one can alternatively compute a \emph{spectral basis} $\{\phi_i(\z)\}_{i\geq 0}$  \citep{LeeIzbickiReg}. This basis is orthonormal with respect to the density $f(\z)$ and adapts to the density's intrinsic geometry.  The expansion coefficients
 are then given  by  $\beta_{i}(\x) =\langle f(.|\x),\phi_{i}\rangle_{f(\z)} = \int_{\Re^d} \phi_{i}(\z) f(\z|\x) f(\z) d\z = \E\left[\phi_{i}(\Z)f(\Z)|\x\right],$ in which case one needs to estimate $f(\Z)$ as well.
 

 \subsection{Connection to Other Methods}
\label{sec::connections}

\textbf{Varying-Coefficient Models.} 
The model $f(z|\x)=\sum_{i \in \mathbb{N}}\beta_{i }(\x)\phi_i(z)$ can be viewed as a {\em fully nonparametric} varying-coefficient model. {\em Varying-coefficient models}~\citep{hastie1993varying} are often seen as semi-parametric models or as extensions of classical linear models, 
 in which a function $\eta$ is modeled as $\eta=\sum_{i=1}^d  \beta_i(\x)u_i$, where $\beta_i(\x)$ are smooth functions of the predictors $\x$, and $u_1,\ldots,u_d$ are other predictors. 
 In our case, we have a fully nonparametric model, because $d \longrightarrow \infty$ and  $(u_i)_{i\geq1}:=\{\phi_i(z)\}_{i\geq1}$ is a basis of $\mathcal{L}^2(\Re)$. 

 \vspace{2mm}
  \vspace{2mm}
\textbf{Spectral Series CDE.}
FlexCode recovers the spectral series conditional density estimator of \citet{IzbickiLeeCDE} if each $\beta_{i}(\x)$ is estimated via a {\em spectral series regression} \citep{LeeIzbickiReg}. Indeed, let  $\{\psi_j\}_j$ be a spectral basis for $\x$, where by construction $\int_\mathcal{X} \! \psi_i(\vec{x})\psi_j(\vec{x})dP(\vec{x})=\delta_{i,j}\overset{\mbox{\tiny{def}}}{=}\I(i=j)$ \citep[Sec.~2]{IzbickiLeeCDE}. In spectral series CDE, one writes the conditional density as 
$f(z|\x)=\sum_{i\geq1}  \sum_{j\geq1}  \beta_{i,j} \phi_{i}(z) \psi_{j}(\x),$
where the coefficients
\begin{align}
\label{eq::betaCDESpec}
\beta_{i,j}= \iint f(z|\vec{x})  \phi_{i}(z) \psi_{j}(\x) \:dP(\vec{x})dz = \E\left[\phi_i(Z)\psi_j(\X)\right].
\end{align} 

Now,  a spectral series regression for $\beta_i(\x)=\E\left[\phi_i(Z)|\x\right]$ is based on the model $\beta_{i}(\x)=\sum_{j\geq1} \gamma^{(i)}_j \psi_j(\x),$
 where
 \begin{align}
 \label{eq::betaCDESpec2}
 \gamma^{(i)}_j &= \int \beta_{i}(\x) \psi_j(\x) \:dP(\vec{x}) = \int \E\left[\phi_i(Z)|\x\right] \psi_j(\x) \:dP(\vec{x}) =  \notag \\ 
 &= \int \E\left[\phi_i(Z)\psi_j(\X) |\x\right] \:dP(\x) = \E\left[\phi_i(Z)\psi_j(\X)\right].
 \end{align}
By inserting $\beta_{i}(\x)$ into Eq.~\ref{eq::FlexCode}, we see that Spectral Series CDE \citep{IzbickiLeeCDE} is a special case of FlexCode.
Henceforth, we will refer to this version of FlexCode as \emph{FlexCode-Spec}.

 \emph{Remark:}  Using similar arguments, one can show that FlexCode recovers the orthogonal series CDE of \citet{efromovich} if each $\beta_{i}(\x)$ is estimated via traditional orthogonal series regression. However, as discussed in  \citet{IzbickiLeeCDE}, traditional series approaches via tensor products quickly become intractable in high dimensions. 
Nevertheless, it is interesting to note that FlexCode forms a very large family of CDE approaches that includes Spectral Series CDE and traditional orthogonal series CDE as special cases.
   
  \vspace{2mm}

\section{Experiments}
\label{sec::appli}

In what follows, we compare the following estimators: 

\begin{itemize}[leftmargin=*]
	\item  FlexCode is our proposed series approach. We implement  six versions of FlexCode, where we use different regression methods to compute the coefficients $\widehat{\beta}_{i }(\x) = \widehat{\E}\left[\phi_i(Z)|\x\right]$ in Eq.~\ref{eq:cdeEst}.
	\emph{FlexCode-SAM} is based on Sparse Additive Models \citep{Ravikumar:EtAl:2009}.\footnote{Sparse additive regression models can be useful even if the true coefficients $\beta_{i }(\x)$ are not additive, because of the curse of dimensionality and the ability of sparse additive models to identify irrelevant coefficients without too restrictive assumptions.} 
	\emph{FlexCode-NN} is based on  Nearest Neighbors regression \citep{Hast:Tibs:Frie:2001}. \emph{FlexCode-Spec} uses  Spectral Series regression  \citep{LeeIzbickiReg} and  is, as shown in Sec.~\ref{sec::connections}, the same as Spectral Series CDE, the conditional density estimator in \citet{IzbickiLeeCDE}. For mixed data types, we implement  \emph{FlexCode-RF}, which estimates the regression functions via random forests \citep{breiman2001random},  and for functional data, we use \textit{FlexCode-fKR}, where the coefficients in the model are estimated via functional kernel regression \citep{ferraty2006nonparametric}. Finally, in Sec.~\ref{sec::distributional},  we illustrate how \textit{FlexCode-SDM} can extend Support Distribution Machines (SDM; \citealt{sutherland2012kernels}) and other distribution regression methods to estimating conditional densities on {\em sample sets}  or groups of vectors. 
	

	\item \textit{KDE} is the kernel density estimator $\widehat{f}(z|\x):=\widehat{f}(z,\x)/\widehat{f}(\x)$, where  $\widehat{f}(z,\x)$ and $\widehat{f}(\x)$ 
	are standard multivariate kernel density estimators. We rescale the data to have the same mean and variance in each direction, and we assume an isotropic Gaussian kernel for both $\x$ and $z$, i.e.,
	$$ \widehat{f}(z|\x) = \frac{\sum_{i=1}^{n} K_{h_x}(\|\x-\X_i\|)  K_{h_z}(z-Z_i)}{ \sum_{i=1}^{n} K_{h_x}(\|\x-\X_i\|) },$$
	where $K_h(t)=h^{-d} K(t/h)$ denotes an isotropic Gaussian kernel with bandwidth $h$ in $d$ dimensions.

	\item \textit{KDE$_{\mbox{\tiny Tree}}$} is the multivariate kernel density estimator $\widehat{f}(z|\x):=\widehat{f}(z,\x)/\widehat{f}(\x)$,
	where the estimators $\widehat{f}(z,\x)$ and $\widehat{f}(\x)$ have a different bandwidth  for each component of $\x$ \citep{Hall2}; i.e.,
	$$ \widehat{f}(z|\x) = \frac{\sum_{i=1}^{n}  K_{h}(\x-\X_i) K_{h_z}(z-Z_i)}{ \sum_{i=1}^{n} K_{h}(\x-\X_i) },$$
	where $K_h(\x-\X_i)=(h_1 \ldots h_d)^{-1}\prod_{j=1}^{d} K\left(\frac{x_j-X_{ij}}{h_j}\right)$ for  data $\X_i=(X_{i1},\ldots, X_{id})$ and  a bandwidth vector $h=(h_1,\ldots, h_d)$.	
	We use the R package \texttt{np} \citep{np}
	with kd-trees and Epanechnikov kernels for computational efficiency \citep{gray2003nonparametric,Holmes}. 
	
	
	\item  \textit{kNN} is a kernel nearest neighbors approach \citep{Lincheng, IzbickiLeeFreeman} to conditional density estimation; it is defined as
	$$\widehat{f}(z|\vec{x})\propto \sum_{j \in \mathcal{N}_k(\vec{x})}   K_{\epsilon}\left(z-Z_j\right),$$ 
	where $\mathcal{N}_k(\vec{x})$ is the set of the $k$ closest neighbors to $\x$ in the training set, and
	$K_{\epsilon}$ is a multivariate (isotropic) Gaussian kernel with bandwidth $\epsilon$. 

	\item     \textit{fkDE}   is a nonparametric conditional density estimator for functional data \citep{quintela2011nonparametric}. It is defined as
		$$ \widehat{f}(z|\x) = \frac{ \frac{1}{h_z}\sum_{i=1}^{n}K\left(\frac{d(x,X_i)}{h_x}\right) K_0\left(\frac{z-Z_i}{h_z}\right) }{ \sum_{i=1}^{n}  K\left(\frac{d(x,X_i)}{h_x}\right)},$$
	where $d$ is a distance measure in the (functional) space of the data, $K$ and $K_0$ are isotropic kernel functions, and $h_x$ and $h_z$ are tuning parameters.
\end{itemize}

Note that for {\em regression}, SAM is designed to work well when there is a small number of relevant covariates, and both Spectral Series Regression and Nearest Neighbors Regression perform well when the covariates exhibit a low intrinsic dimensionality. To our knowledge, \textit{KDE$_{\mbox{\tiny Tree}}$} is  the only CDE method that can handle mixed data types.

\subsection{Toy Examples} 

By simulation, we create toy versions of common scenarios with different structures in data and different types of data. We use 700 data points for training, 150 for validation and 150 for testing the methods. Each simulation is repeated 200 times. 

\subsubsection*{\underline{Different structures in data.}}

\begin{itemize}
	\item \textbf{Irrelevant Covariates.}
	In this example, we generate data according to
	$Z|\x \sim N(x_1,0.5),$
	where $\X=(X_1,\ldots,X_d) \sim N(\vec{0},I_d)$, that is, only the first covariate influences the response. 
	\vspace{2mm}
	
	\item \textbf{Data on Manifold.}
	Here we let
	$Z|\x \sim N(\theta(\x),0.5),$
	where $\x=(x_1,\ldots,x_d)$ lie on a unit circle
	embedded in a $D$-dimensional space, and $\theta(\x)$ is the angle corresponding to the position of $\x$. For simplicity, we assume that the data are uniformly distributed on the manifold; i.e., we let $\theta(\x) \sim Unif(0,2\pi)$. 
	\vspace{2mm}
	
	\item \textbf{Non-Sparse Data.}
	Finally, we consider data with no sparse (low-dimensional) structure. We assume
	$Z|\x \sim N(\overline{\x},0.5),$
	where $\X=(X_1,\ldots,X_d) \sim N(\vec{0},I_d)$. 

\end{itemize}
\vspace{2mm}

\subsubsection*{\underline{Different types of data.}}
\begin{itemize}
	\item \textbf{Mixed Data Types.} Few existing CDE methods can handle mixed data types; the only other method the authors are aware of is \textit{KDE$_{\mbox{\tiny Tree}}$}. For our study, we generate mixed categorical and continuous data, where the categorical covariates $(X_1,\ldots,X_{D/2})$ are i.i.d. 
	$Unif\{c_1,c_2,c_3,c_4,c_5\}$, and the continuous covariates $(X_{D/2+1},\ldots,X_{D})$
	are i.i.d. $N(0,1)$. The response is given by 
	\[ Z|\x \sim
	\begin{cases}
	N(x_{D/2+1},0.5)    & \quad \text{if } x_{1}\in \{c_1,c_2\} \\
	10+ 2N(x_{D/2+2},0.5)   & \quad \text{if } x_{1}\in \{c_3,c_4,c_5\}\\
	\end{cases}
	\]
	\vspace{2mm}

	\item \textbf{Functional Data.} We also consider spectrometric data for finely chopped pieces of meat. These high-resolution spectra are available\footnote{\url{http://lib.stat.cmu.edu/datasets/tecator}; the original data source is Tecator AB} as a benchmark for functional regression models (see, e.g.,  
	\citet{ferraty2007nonparametric}), where the task is to
	predict the fat content of a meat sample on the basis of its near infrared absorbance spectrum. In our study, we use 215 samples to estimate conditional densities. 
	The covariates are spectra of light absorbance as functions of the
	wavelength, and the response is the fat content of a piece of meat. 
		We compare the functional kernel density estimator (\textit{fKDE}) with a FlexCode approach  (\textit{FlexCode-fKR}), where the coefficients in the model are estimated via functional kernel regression \citep{ferraty2006nonparametric}. 
	We follow \citet{ferraty2007nonparametric} and implement both methods with the kernel function $K(u)=1-u^2$ and the $\mathcal{L}^2(\Re)$ norm between the second derivatives of the spectra as a distance measure. We use 70\% of the data points for training, 15\% for validation and 15\% for testing; the experiment is repeated 100 times by randomly splitting the data.

		
\end{itemize}


\subsubsection{Results}

Figures~\ref{fig::SimExamples}-\ref{fig::SimExamplesLargeD} show the results for the toy data. Our main observations are:
\begin{itemize} 
	\item\textbf{Irrelevant Covariates.} In terms of estimated loss  (Fig.~\ref{fig::SimExamples}, top left),  both \textit{FlexCode-SAM} and  \textit{KDE$_{\mbox{\tiny Tree}}$}  outperform the other methods. 
	However, 
	in terms of computational time (Fig.~\ref{fig::SimExamples}, bottom left), \textit{FlexCode-SAM}  is clearly faster than \textit{KDE$_{\mbox{\tiny Tree}}$} as the dimension $D$ of the data grows.
	%
	When $D=17$, each  fit with \textit{KDE$_{\mbox{\tiny Tree}}$}  already takes an average of 240 seconds (4 minutes) on an Intel i7-4800MQ CPU 2.70GHz processor, compared to 22 seconds for \textit{FlexCode-SAM}.  Fig.~\ref{fig::SimExamplesLargeD}, left, shows that the loss of \textit{FlexCode-SAM} remains the same even for large $D \sim 1000$, although fitting the estimator becomes computationally more challenging in high dimensions. Nevertheless, fitting \textit{KDE$_{\mbox{\tiny Tree}}$}  would be unfeasible for $D>50$.
	
	\item\textbf{Data on Manifold.} \textit{FlexCode-Spec}
	has the best statistical performance, followed by
	\textit{FlexCode-NN} and \textit{KDE$_{\mbox{\tiny Tree}}$} (Fig.~\ref{fig::SimExamples}, top center).
	As before, the computational time  of \textit{KDE$_{\mbox{\tiny Tree}}$} increases rapidly with the dimension (Fig.~\ref{fig::SimExamples}, center bottom). For these data, \textit{FlexCode-SAM} is slow as well even for moderate $D$, perhaps because SAM cannot find sparse representations of the regression functions.  On the other hand (see Fig.~\ref{fig::SimExamplesLargeD}, {\em right}), \textit{FlexCode-Spec} has a computational time that is almost constant as a function of $D$ and the statistical performance remains the same even for large $D$.
	The latter result is consistent with our previous findings that spectral series adapt to the intrinsic dimension of the data \citep{izbickiLeeSchafer,IzbickiLeeCDE,LeeIzbickiReg}.

	
	\item\textbf{Non-Sparse Data.} For this example,
	\textit{FlexCode-Spec} and \textit{FlexCode-SAM}
	are the best estimators.		
	
	\item\textbf{Mixed Data Types.}  \textit{FlexCode-RF} yields better results than 
	\textit{KDE$_{\mbox{\tiny Tree}}$} (its only competitor in this setting)
	both in terms of estimated loss and computational time; see Fig.~\ref{fig::MixedData}. The computational advantage is especially obvious for larger values of $D$. 
		When the dimension $D=56$, each  fit  of \textit{KDE$_{\mbox{\tiny Tree}}$} takes an average of 2250 seconds ($\approx$ 37 minutes) on an Intel i7-4800MQ CPU 2.70GHz processor, compared to 304  seconds ($\approx$ 5 minutes) for \textit{FlexCoDR-RF}.
		
		\item\textbf{Functional Data.}  		FlexCode via Functional kernel regression improves upon the results of the  traditional Functional kernel density estimator with an estimated loss of -2.78 (0.07) instead of -2.08 (0.03). 

	\end{itemize}
	
	\begin{figure}[H]
		\centering
		\subfloat{  \includegraphics[page=1,scale=0.27]{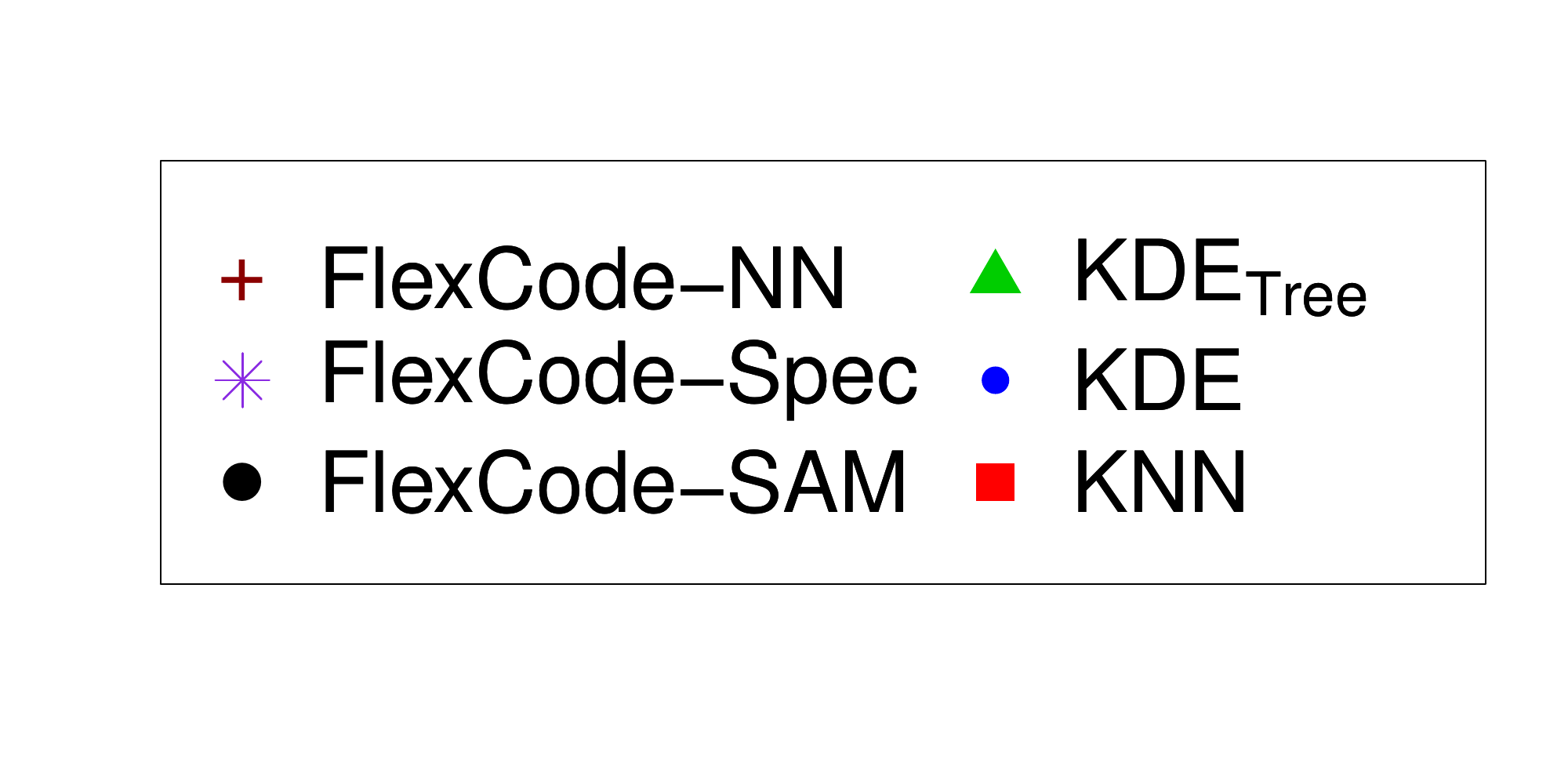}} 
		\\[-8.0mm] 
		\subfloat{  \includegraphics[page=1,scale=0.27]{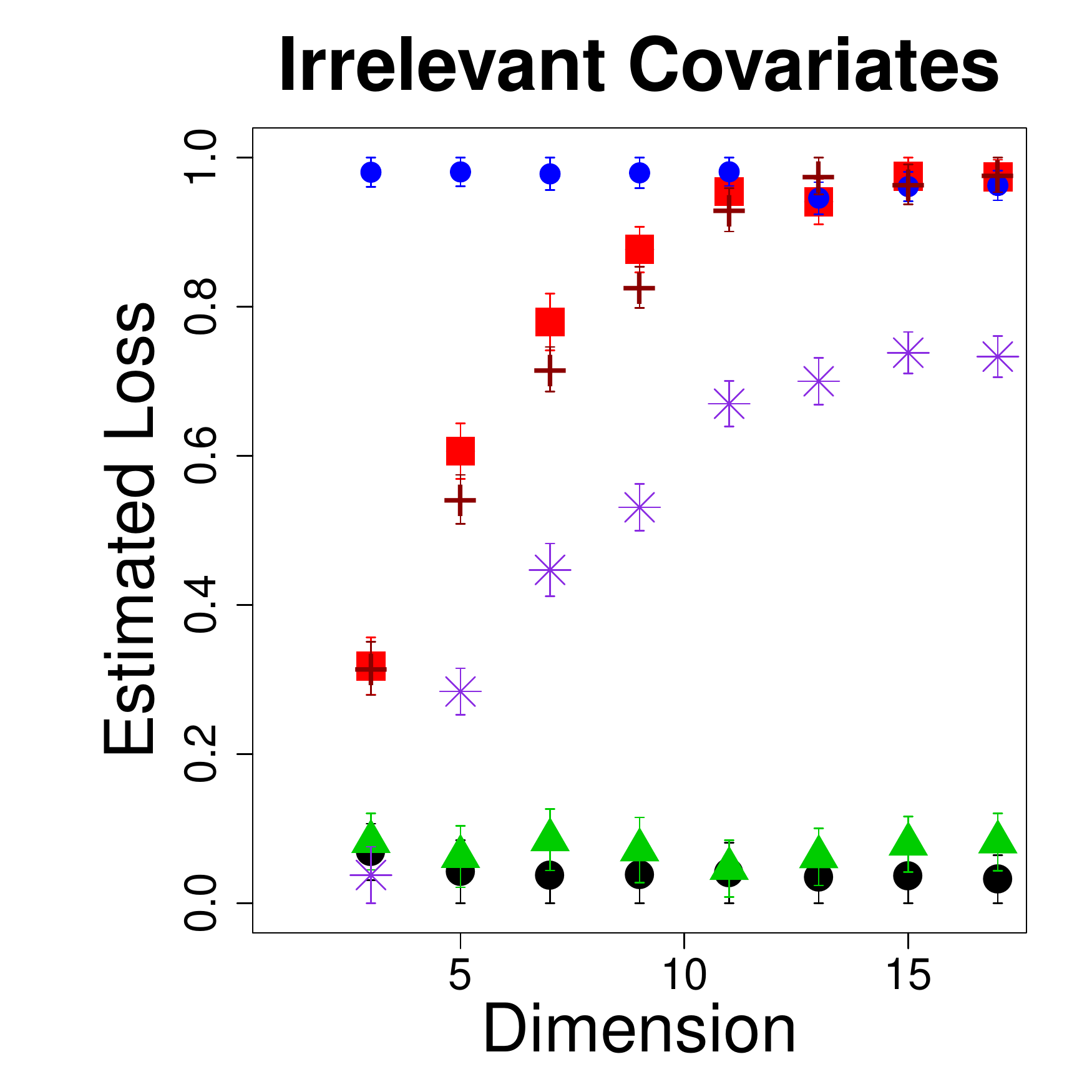}} 
		\subfloat{  \includegraphics[page=5,scale=0.27]{analysisRefereesForPaper.pdf}} 
		\subfloat{  \includegraphics[page=3,scale=0.27]{analysisRefereesForPaper.pdf}} 
		\\[-4.0mm] 
		\subfloat{  \includegraphics[page=2,scale=0.27]{analysisRefereesForPaper.pdf}} 
		\subfloat{  \includegraphics[page=6,scale=0.27]{analysisRefereesForPaper.pdf}} 
		\subfloat{  \includegraphics[page=4,scale=0.27]{analysisRefereesForPaper.pdf}} 
		\vspace{-3mm}
		\caption{\footnotesize  Examples with different structures in data. {\em Top row:} Estimated loss as a function of the dimension $D$. {\em Bottom row:} Computational time as a function of $D$. With a properly chosen regression method, \textit{FlexCode} performs better than the other estimators (\textit{KDE$_{\mbox{\tiny Tree}}$}, \textit{KDE}, and \textit{kNN}).
		}
		\label{fig::SimExamples}
	\end{figure}

	\begin{figure}[H]
		\centering
		\subfloat{  \includegraphics[page=9,scale=0.27]{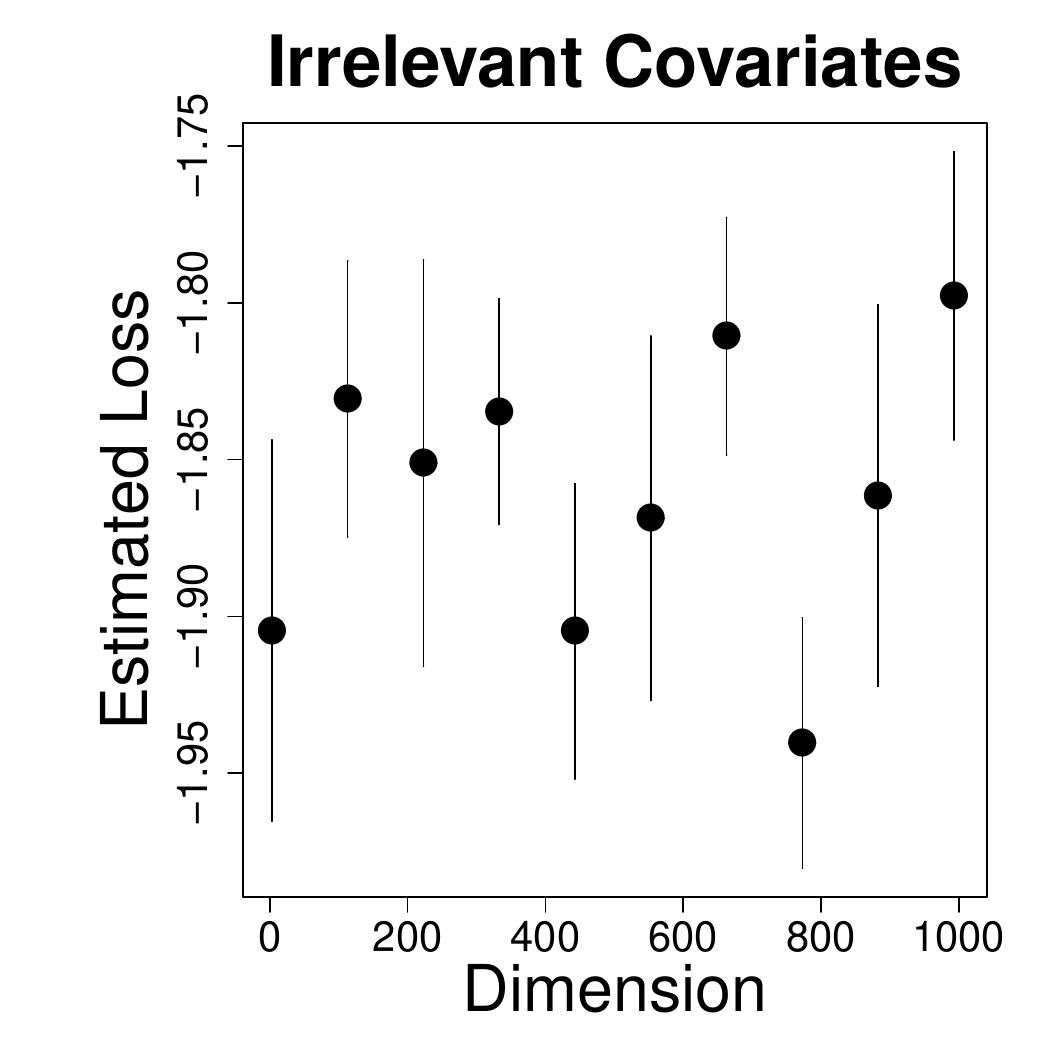}} 
		\subfloat{  \includegraphics[page=11,scale=0.27]{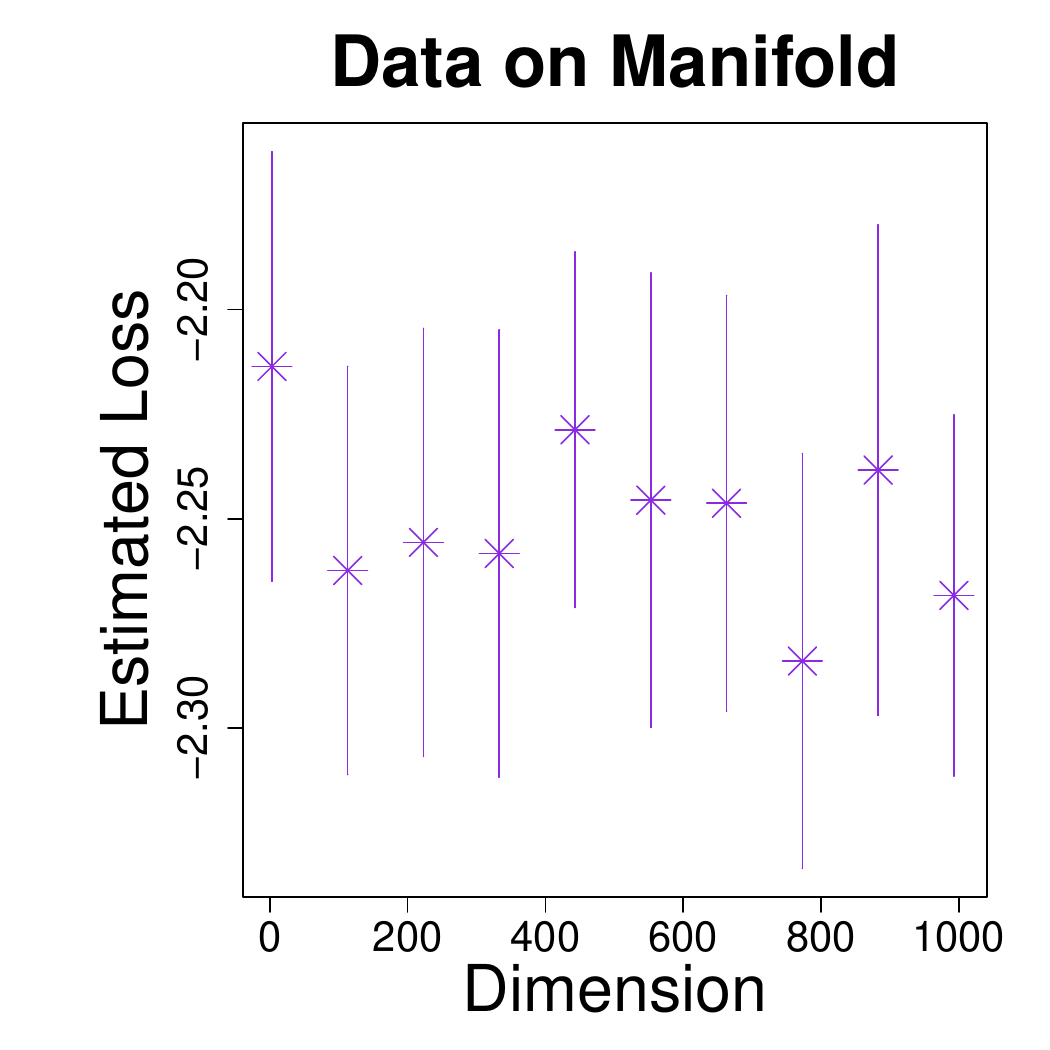}} 	\\[-4.0mm] 
		\subfloat{  \includegraphics[page=10,scale=0.27]{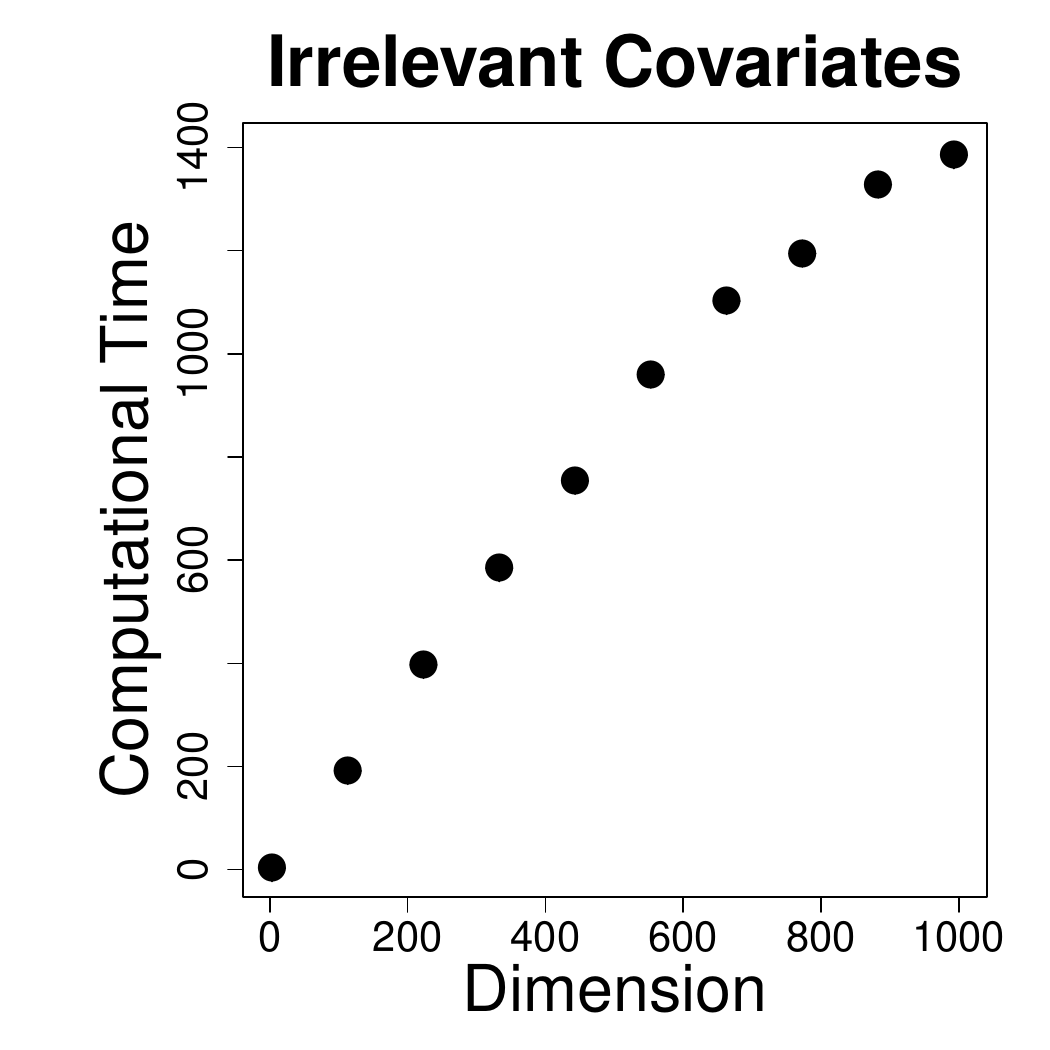}} 
		\subfloat{  \includegraphics[page=12,scale=0.27]{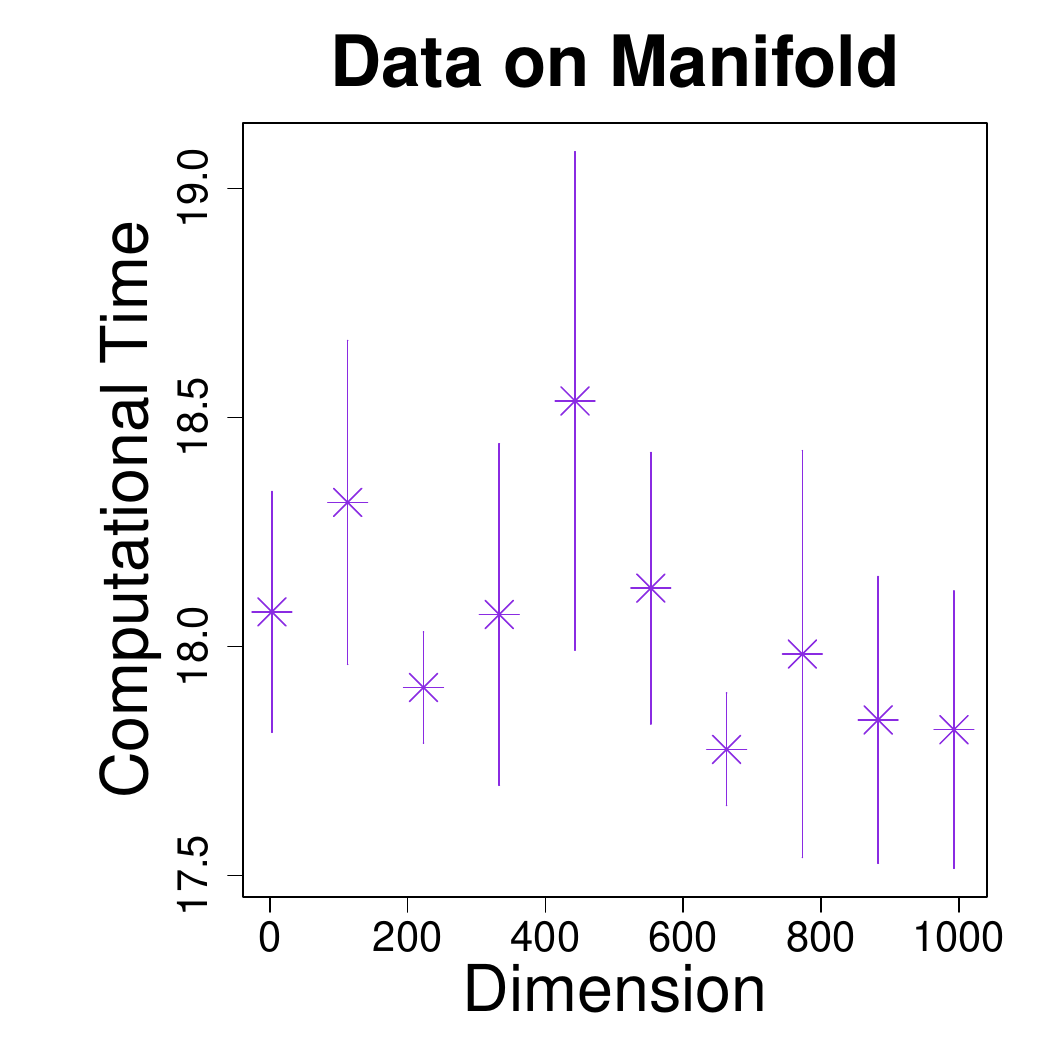}} 
		\vspace{-3mm}
		\caption{\footnotesize Different structures in data for large values of $D$. Estimated loss ({\em top row}) and computational time ({\em bottom row}) of FlexCode in the two settings ``Irrelevant Covariates'' ({\em left}; implemented with \textit{FlexCode-SAM}) and ``Data on Manifold''  ({\em right}; implemented with \textit{FlexCode-Spec}). These two estimators yield the best results in Fig.~\ref{fig::SimExamples}; here we see their behavior in higher dimensions. 
		 }
		\label{fig::SimExamplesLargeD}
	\end{figure}
	

	\begin{figure}[H]
		\centering
		\subfloat{ \includegraphics[page=2,scale=0.2]{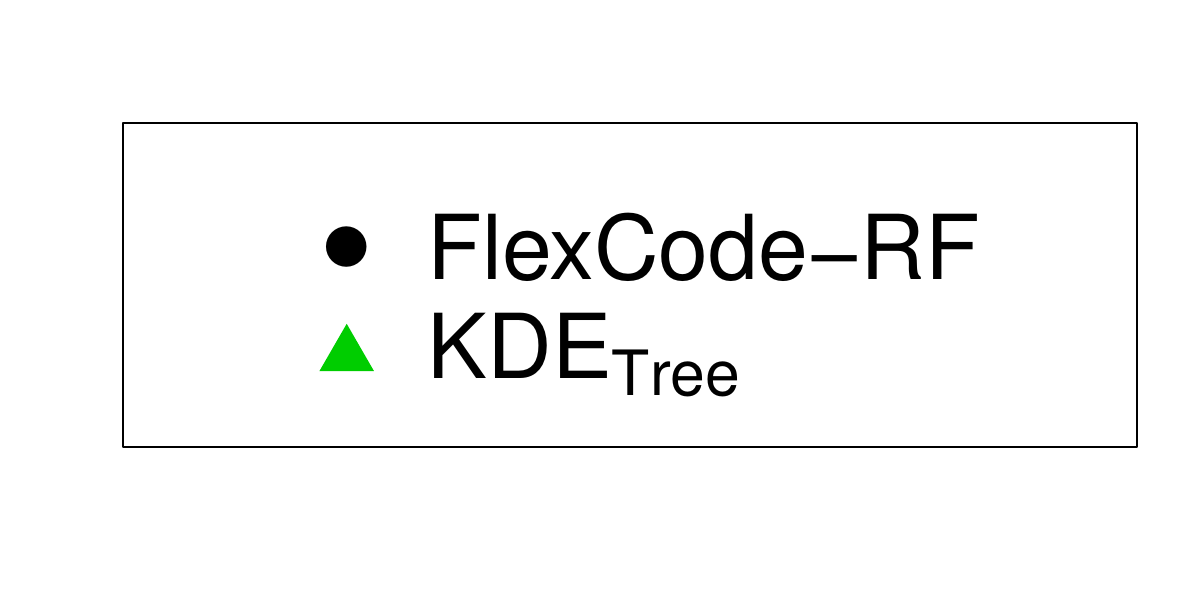}} 
		\\[-8.0mm] 
		\subfloat{  \includegraphics[page=7,scale=0.27]{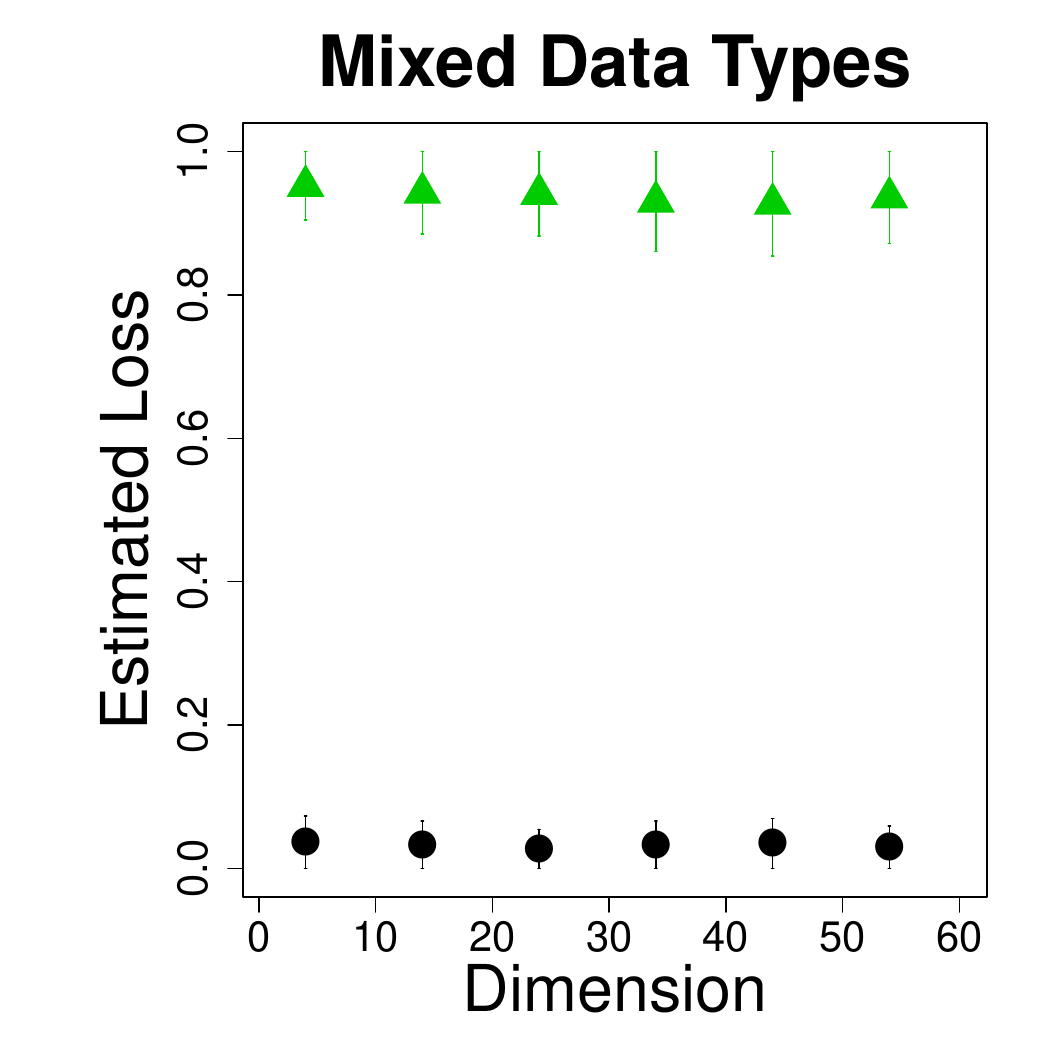}} 
		\subfloat{  \includegraphics[page=8,scale=0.27]{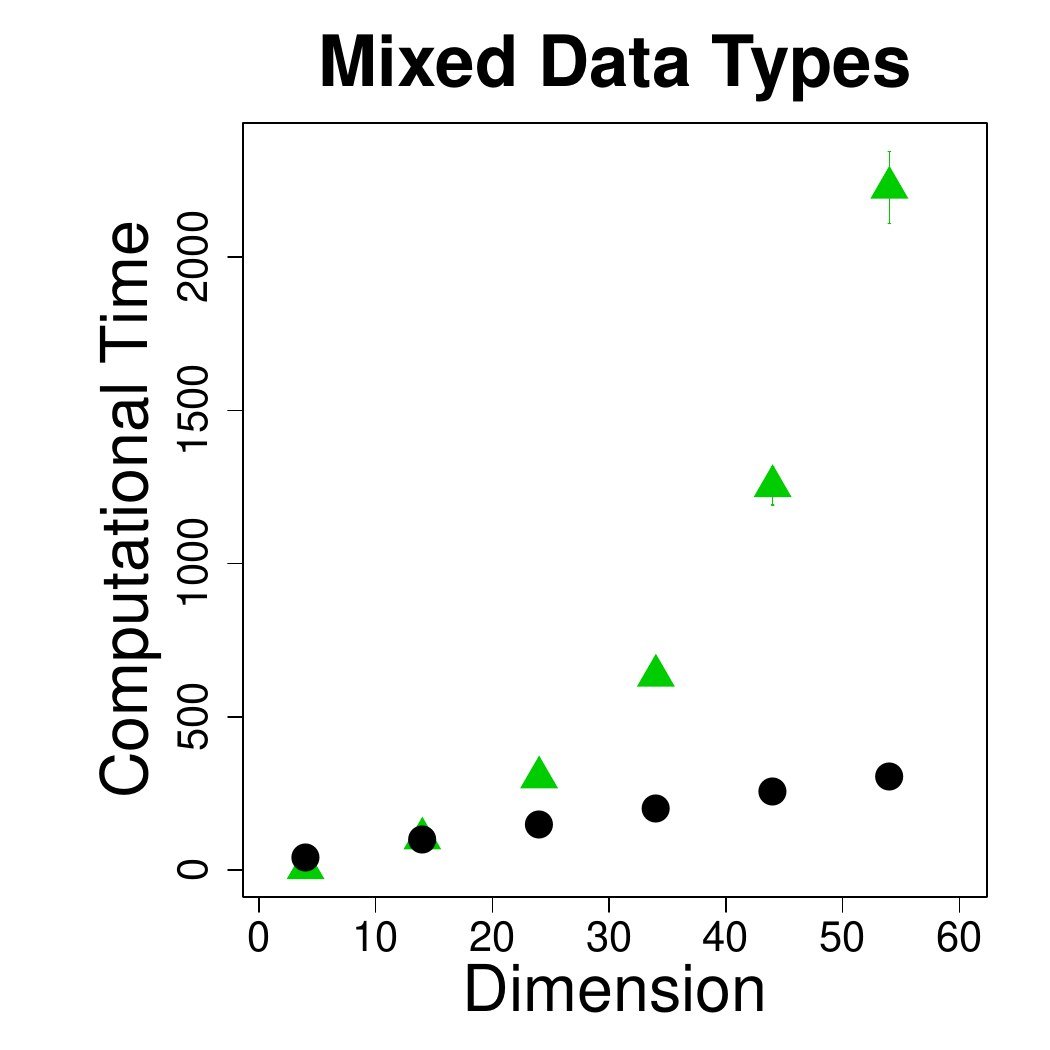}} 
		\vspace{-3mm}
		\caption{\footnotesize Example with mixed data.  Estimated loss ({\em left}) and computational time ({\em right}) 
			for FlexCode via Random Forests (\textit{FlexCode-RF}) and \textit{KDE$_{\mbox{\tiny Tree}}$}.  Few conditional estimators can handle  covariates with mixed data types, but FlexCode is flexible enough to adapt to this setting.}
		\label{fig::MixedData}
	\end{figure}
	

	
	\subsection{Photometric Redshift Estimation}
	\label{sec::photoz}
	

 Our first application is photometric redshift estimation. Redshift (a proxy for a galaxy's distance from the Earth) is a key quantity for inferring cosmological model parameters. Redshift can be estimated with high precision via spectroscopy but the resource considerations of large-scale sky surveys call for {\em photometry} -- a much faster measuring technique, where the radiation from an astronomical objects is generally coarsely recorded via $\sim$5-10 broad-band filters. In photometric redshift estimation, the goal is to 
estimate the redshift $z$
of a galaxy based on its observed photometric covariates $\vec{x}$, using a sample of galaxies with spectroscopically confirmed redshifts. Because of degeneracies  (two galaxies with different redshifts can have similar photometric signatures) and because of complicated observational noise, probability densities of the form $f(z|\x)$ better describe the relationship between $\x$ and $z$ than the regression $\E(z|\x)$ does.


In this example, we test our CDE methods on  $n=752$  galaxies from {\em COSMOS}, with $D=37$ covariates derived from a variety of photometric bands (these data were obtained
from T.~Dahlen 2013, private communication; see \citealt{IzbickiLeeFreeman} for additional details).  Figure \ref{fig::photoZLoss} summarizes the results. All versions
	of FlexCode improve upon the traditional estimators. The best performance
	is achieved for FlexCode via Sparse Additive Models (\textit{FlexCode-SAM}), which indicates that
	only a subset of the 37 covariates are relevant for redshift estimation; for these data, \textit{FlexCode-SAM} selected $\approx$ 18  variables in each regression, and three out of the 37 covariates were present in more than 75\%
	of the regressions.

%
%
%
%
%

		\begin{figure}[H]
			\centering
			\subfloat{  \includegraphics[page=1,scale=0.27]{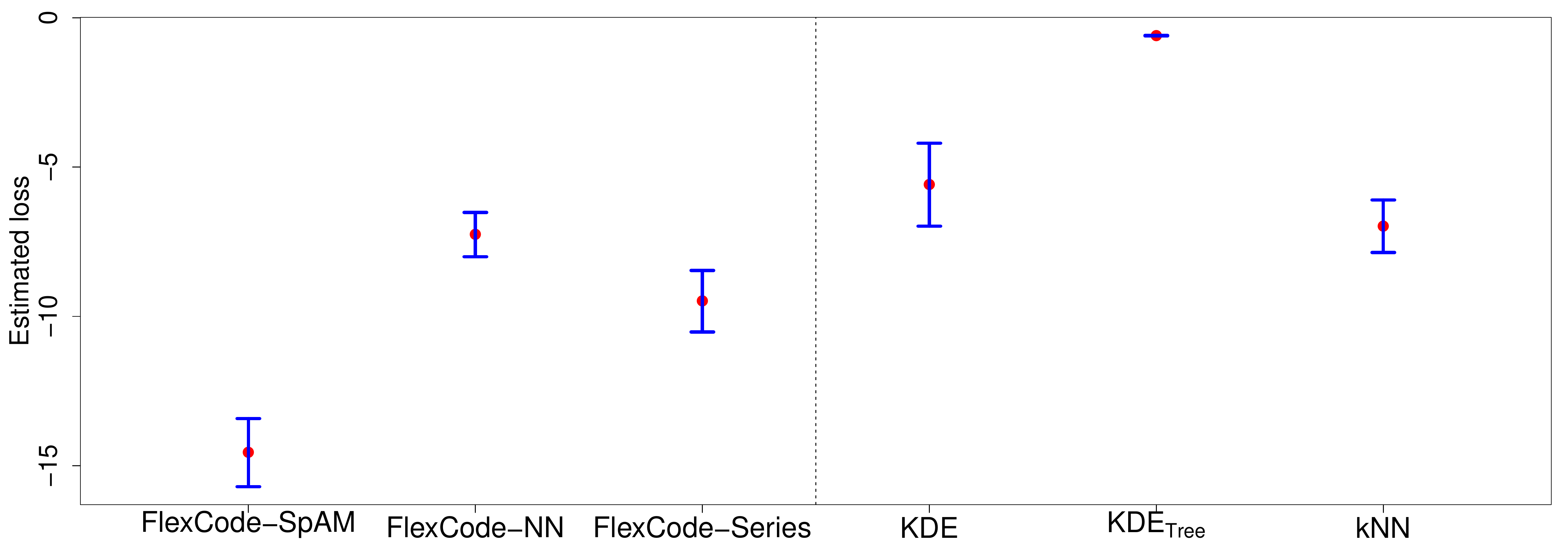}} 
			\vspace{-3mm}
			\caption{\footnotesize Estimated losses of conditional density estimators for photometric redshift prediction. All versions
	of FlexCode (to the left) improve upon the traditional estimators (to the right)}
			\label{fig::photoZLoss}
		\end{figure}

	\subsection{Twitter Data}
	\label{sec::twitter}

	Twitter is a social network where each user is
	able to post a small text (a tweet) containing at most 140
	characters. 
	Information about the location of the post is available upon user permission, but only a few users allow this information to be publicly shared.
	Here we use samples with known  locations to estimate the location of tweets
	where this information has not been shared publicly.	
	
	Note that most literature on the topic concerns creating  {\em point estimates}
	for locations  (see, e.g., \citealt{Rodrigues2015} 
	and references therein). In this work, we estimate
	the \emph{full conditional distribution} of latitude and longitude
	given the content of the tweet; that is,
	we estimate
	$f(\z|\x)$, where $\x$ are covariates extracted
	from the tweets and
	$\z=(z_1,z_2)$ is the pair latitude/longitude.
	
	Our data set   
	 contains $\approx 8000$ tweets
	in the USA from July 2015 with the word ``beach". We extract 500 covariates via a bag-of-words method with the most frequent unigrams and bigrams \citep{manning2008introduction}.
	As we only expect a few of the 500 covariates to be relevant to locating the tweets, we implement FlexCode via sparse additive models.
	Figure \ref{fig::twitter} shows two examples of estimated
	densities; see Supplementary material for additional examples. To our knowledge, 
	no other fully nonparametric conditional density estimation method can be directly applied to  these types of data where there are many irrelevant variables.

	Moreover, because \textit{FlexCode-SAM} is based on sparse additive models, we can find out which covariates are most relevant for predicting location.
	For the example in Fig.\ref{fig::twitter}, left,
	the expressions 
	``beachin",``boardwalk", and ``daytona"
	are included in at least 33\% of the estimated regression functions. 
	For the example to the right, the relevant covariates are 
	``long beach", ``island",
	``long", and ``haven". 

		\subsection{From Distribution Regression to ``Distribution CDE'': 
		Estimating the Mass of a Galaxy Cluster from Sample Sets of Galaxy Velocities}
		\label{sec::distributional}

		Distribution regression and classification is a recent emerging field of machine learning. Instead of treating individual data points (or feature vectors) as covariates, these methods operate on {\em sample sets}, where each set is a sample from some underlying feature distribution; see \citet{sutherland2012kernels} and references within.  Here we show that FlexCode extends to sample sets as well; 
		our application is estimation of the mass of a galaxy cluster given the line-of-sight velocities of the galaxies in the  cluster.
				
			Galaxy clusters, the most massive gravitationally bound systems in the Universe, can contain up to $\sim$1000 galaxies. These structures are a rich source of information on astrophysical processes 
		and cosmological parameters,  
		but to use galaxy clusters as cosmological probes one needs to accurately measure their masses.
		A standard approach is to
		employ the classical virial theorem and directly relate the mass of a cluster to  the line-of-sight (LOS) galaxy velocity dispersion, i.e., the variance of the measured galaxy velocities in the cluster  \citep{evrard2008virial}. Recently, \citet{ntampaka2015machine} and \citet{ntampaka2015dynamical} have shown that one can significantly improve such mass predictions by taking advantage of the entire LOS {\em velocity distribution} of galaxies instead of only the dispersion (i.e., a summary of the distribution). Here we show that FlexCode can further improve these results.

		The general set-up is that we observe data of the form
		$(\x^{(1)}_1,
		\ldots,\x_{1}^{(J_1)},z_1),\ldots,(\x_{I}^{(1)},
		\ldots,\x_{I}^{(J_I)},z_I)$, where $z_i$
		is the mass of the $i$-th cluster for $i=1,\ldots,I$; and
		$\x_{i}^{(j)}$ is a vector of galaxy observables (such as LOS velocity and the projected distance from the cluster center) for the 
		$j$-th galaxy in the $i$-th cluster.
		Note that different clusters $i$ contain different numbers $J_i$ of galaxies. The key idea behind Support Distribution Machines (SDMs; proposed for this application by \citealt{ntampaka2015machine}) as well as other ``distribution regression'' methods \citep{sutherland2012kernels}, is to treat each sequence
		$\x_{i}^{(1)},
		\ldots,\x_{i}^{(J_i)}$ as a sample from
		a probability distribution $p_i$, and to construct an appropriate kernel matrix
		 on these sample sets. The task is then to predict a scalar ($z_i$) from a distribution ($p_i$) by estimating $\mathbb{E}[Z|p]$.
		Here we show how FlexCode extends regression on distributions to {\em conditional density estimation on distributions}; i.e., instead of providing a point estimate (and standard error) of the mass of a galaxy cluster, we estimate the full probability density $f(z|p)$ of the unknown mass of a galaxy cluster given galaxy observables. In our application, the response $z_i$ is the logarithm of the cluster mass (log M) and the observables $\{x_{i}^{j}\}_{j=1}^{J_i}$ are scalar quantities that represent the absolute values of galaxy velocities along {\em one} line-of-sight.		
	\comment{	The challenge in this particular problem is that 
		the observed samples
		are given by
		$(\x^{(1)}_1,
		\ldots,\x_{1}^{(I_1)},z_1),\ldots,(\x_{n}^{(1)},
		\ldots,\x_{n}^{(I_n)},z_n)$, where $z_i$
		is the log mass of the $i$-th cluster and
		$\x_{i}^{(j)}$ is the vector of observed
		absolute values of the velocities of the 
		$j$-th galaxy of the $i$-th cluster
		along each line-of-sight. Notice different  clusters may have
		a different number of galaxies on them. It is therefore not 
		possible to use standard conditional density estimators; in fact even the regression task is challenging here. 
		The only regression method we are aware of that has been applied
		to this problem is the so-called \emph{Support
			Distribution Machine} \citep{ntampaka2015machine}.
		The key idea is to treat each sequence
		$\x^{(i)}_1,
		\ldots,\x_{i}^{(I_i)}$ as a sample from
		a distribution $p_i$, and build an appropriate kernel matrix
		based on this. Here we show how one can easily adapt this idea
		for conditional density estimation via FlexCode. 
		All the approaches we compare here are based on the similarity matrix build by \citet{ntampaka2015machine}, which we describe in details in the sequence.}

		Like \citet{ntampaka2015machine}, we use the Kullback-Leibler (KL) divergence to measure similarity between pairs of velocity distributions, and we estimate the divergence from the observed galaxy velocities with the estimator from \citet{wang2006nearest}. The details are as follows: Let $p_A$ and $p_B$ denote velocity distributions 
		for clusters $A$ and $B$, respectively. Define
	 the kernel $k(p_A,p_B)=\exp{(-\mbox{KL}(p_A, p_B)/\sigma^2)},$
		where $\mbox{KL}(p_A, p_B)$ is the Kullback-Leibler divergence between $p_A$ and $p_B$.
		We estimate the KL divergence via Wang et al's $k$ nearest neighbors method for $k=2$. That is, let $X_A$ denote the set of LOS velocities associated with the
		$n$ galaxies of cluster $A$, 
		and let 
		$X_B$ denote the set of velocities associated with the
		$m$ galaxies of cluster $B$. The estimated KL divergence from $p_A$ to $p_B$ is given by
		$$\mbox{KL}_{n,m}(X_A,X_B)=\frac{d}{n}\sum_{i=1}^n \log \frac{\nu_k(i)}{\rho_k(i)}+\log \frac{m}{n-1},$$
		where
		$\nu_k(i)$ is the Euclidean distance from the covariates (in this case, the LOS velocity) of the $i$-th galaxy in $X_A$
		to its $k$-th nearest neighbor in $X_B$, $\rho_k(i)$ is the Euclidean distance from the covariates (the LOS velocity) of the $i$-th galaxy in $X_A$
		to its $k$-th nearest neighbor in $X_A$, and $d$ is the number of galaxy observables (in this example, $d=1$).
		As the computed kernel matrix $k(X_A,X_B)=\exp{(-\mbox{KL}_{n,m}(X_A,X_B)/\sigma^2)}$ may not be 
		positive semi-definite (PSD), we project the matrix to the closest
		PSD matrix in Frobenius norm \citep{higham2002computing}.

		Using the PSD kernel matrix, we then estimate the conditional density $f(z|p)$. 
		We compare four approaches to conditional density estimation on distributions, which as in the rest of the paper use a Fourier basis in $z$;
		\begin{itemize}
		\item  Functional KDE: the functional kernel density estimator \citep{quintela2011nonparametric},
						\item FlexCode-NN: FlexCode with Nearest Neighbors regression,
			\item FlexCode-Spec: FlexCode with Spectral Series regression,
			\item FlexCode-SDM: FlexCode with SDM regression.
		\end{itemize}
		\vspace{2mm}
 In the experiments, we also include a FlexCode estimator that use a wavelet basis in $z$;
		\begin{itemize}
		\item FlexCode$_W$-SDM: FlexCode with SDM regression in $x$, and Daubechies wavelets with 3 vanishing moments in $z$.
		\end{itemize}

	 Our data consist of simulations of $n=5028$ unique galaxy clusters with minimum mass of $1 \times 10^{14} \ M_\odot h^{-1}$;
	see \citet{ntampaka2015machine} for details. All four methods above are based on the same distance computation $\mbox{KL}_{n,m}(X_A,X_B)$ with $k=2$, and we use data splitting and the loss (\ref{lossEmpirical}) for selecting tuning parameters. For simplicity, we only consider
	one LOS for each cluster (the $x$-axis LOS in the catalog).
	

		It is clear from Table \ref{tab::cluster} that the FlexCode-SDM  and FlexCode$_W$-SDM 
		estimates of conditional density are more accurate than the results from any other method. The coverage plots (see Appendix~\ref{sec::Appendix_diagnostics} for the definition) in the bottom panel of Fig~\ref{fig::clusterCoverage} also verify that these density estimates fit the observed data well. 
	\vspace{0.2in}

	  \begin{table}[H]
	  	\caption{\footnotesize Estimated losses of conditional density estimates of galaxy cluster mass.}
		\vspace{-0.3in}
	  	\begin{center}
	  		\begin{tabular}{l*{6}{c}ccc}   \textit{Functional KDE} & \textit{FlexCode-NN} & \textit{FlexCode-Spec} & \textit{FlexCode-SDM} & \textit{FlexCode$_W$-SDM} \\ 
	  			\hline  
				-0.98 (0.02) &   -1.60 (0.05)  & -1.86 (0.04)  & \textbf{-2.46 (0.09)} & \textbf{-2.71 (0.09)}
	  		\end{tabular}
	  		\label{tab::cluster}
	  	\end{center}
	  \end{table}
		The top left panel of Figure \ref{fig::clusterCoverage} shows examples of density estimates from FlexCode$_W$-SDM for 16 randomly chosen clusters.  Several of these distributions are bimodal, in which case regression estimates are not very informative. This can be further illustrated by Fig.~\ref{fig::clusterMultimodal}. The left panel shows a scatter plot of the observed log masses  versus the estimated conditional mean $\widehat{\E}[Z|p]:=\int \! z \widehat{f}(z|p)dz$ for unimodal versus multimodal cases. The right panel shows a boxplot of  the absolute fractional mass error $|\varepsilon|$  for the two populations; 
	 the fractional mass error $\varepsilon$ is defined as  \citep{ntampaka2015machine}
	 $$ \varepsilon = \frac{M_{\rm pred}-M}{M},$$
	where  $M$ is the observed cluster mass and $M_{\rm pred}$ is the predicted cluster mass. Much of the scatter can indeed be attributed to multimodal densities and non-standard prediction settings.

%

\begin{figure}[H]
	\centering
	\subfloat[FlexCode$_W$-SDM]{  \includegraphics[page=1,scale=0.33]{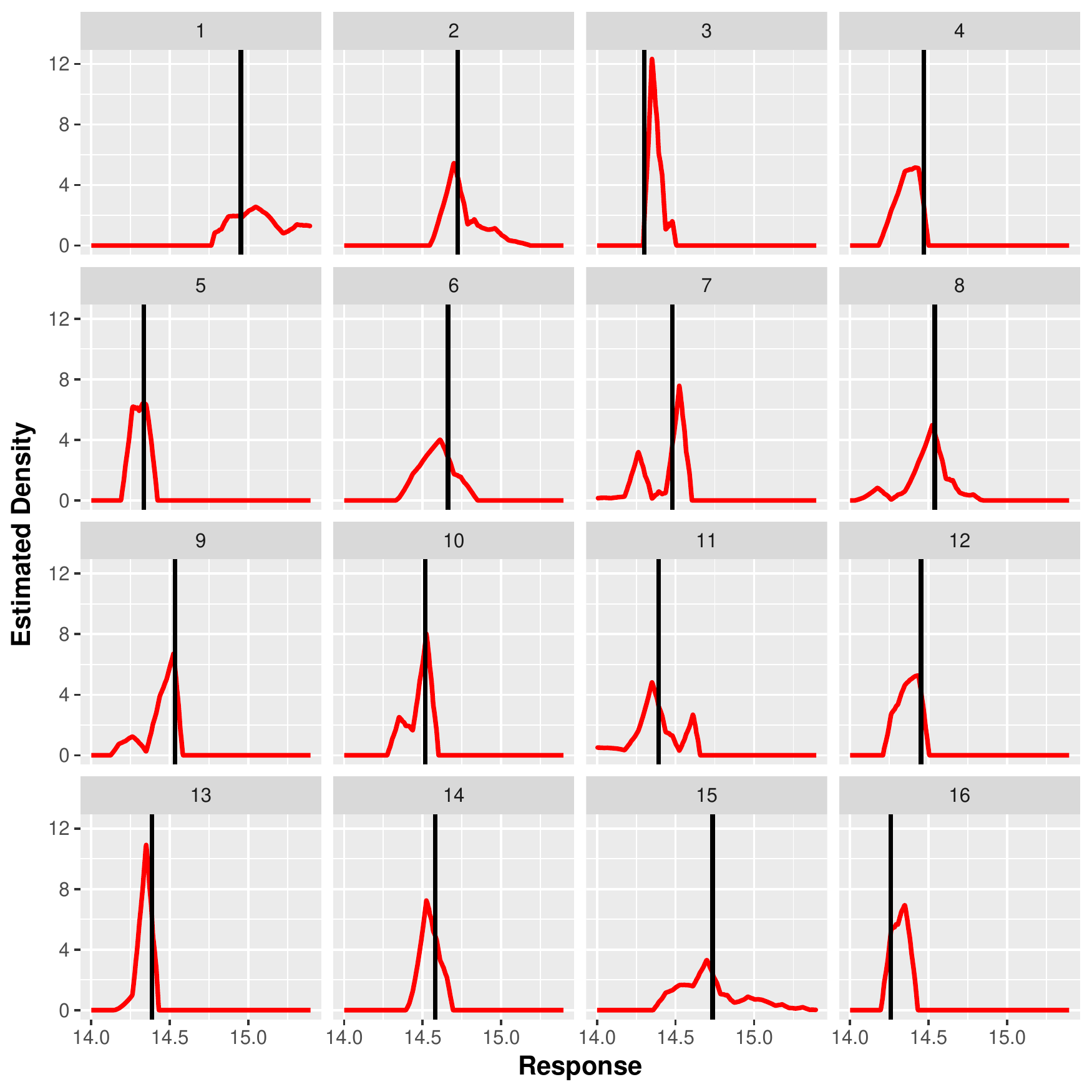}}  
	\subfloat[FlexCode$_W$-SDM]{  \includegraphics[page=2,scale=0.35]{coverageDensitiesWavelets.pdf}}  \\
	\subfloat[FlexCode-SDM]{  \includegraphics[page=11,scale=0.33]{SDMCDEGaussSDMFlexCodeMichelle.pdf}} 					
	\subfloat[FlexCode$_W$-SDM]{  \includegraphics[page=11,scale=0.33]{SDMCDEWavelets.pdf}} 					
	\vspace{-3mm}
	\caption{\footnotesize   {\em Top left}: Estimated probability distributions of the log mass (``Response'') of 16 randomly chosen clusters; the vertical lines show the true values, and the red curves are computed using FlexCode$_W$-SDM. Many of these densities are multimodal and asymmetric, indicating that standard prediction approaches may not accurately model the uncertainty in the mass estimates.
	  {\em Top right}: 95\% highest predictive density (HPD) regions for
	the same 16 clusters derived from the FlexCode$_W$-SDM estimates; the dots show the true values.
	{\em Bottom}: Coverage plots of the density estimates from FlexCode-SDM and FlexCode$_W$-SDM for the entire mock cluster catalog. The plots show that these density estimates fit the observed data well.}
	\label{fig::clusterCoverage}
\end{figure}

					\begin{figure}[H]
						\centering
						\subfloat[FlexCode$_W$-SDM ]{  \includegraphics[page=8,scale=0.29]{SDMCDEWavelets.pdf}} \hspace{2mm}
						\subfloat[FlexCode$_W$-SDM ]{ \includegraphics[page=9,scale=0.29]{SDMCDEWavelets.pdf}} \hspace{2mm}
						\subfloat[SDM Regression]{  \includegraphics[page=1,scale=0.27]{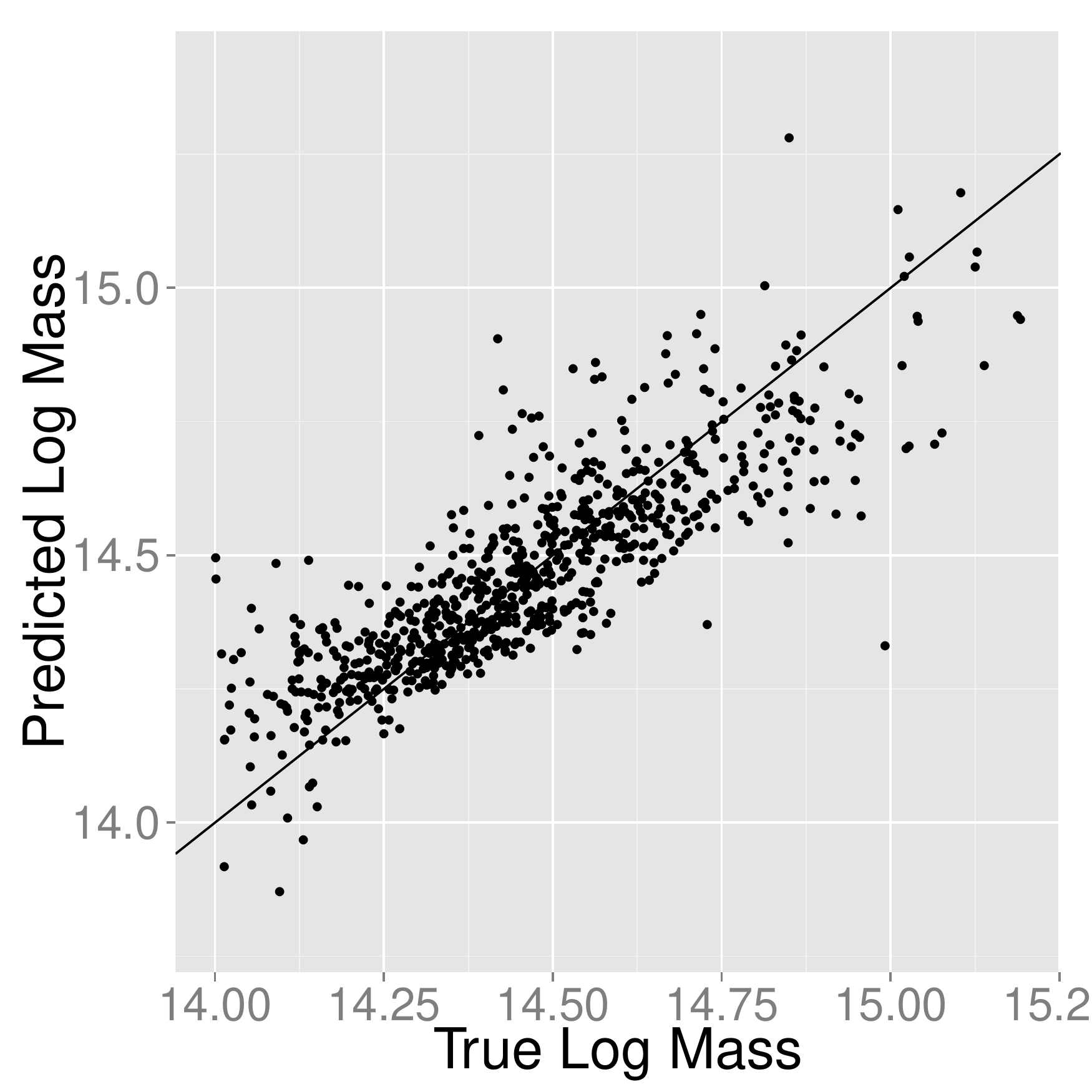}}
						\vspace{-3mm}
						\caption{\footnotesize {\em Left:} 
 Scatter plot of the predicted versus the true log masses for FlexCode$_W$-SDM when taking the conditional mean  $\widehat{\E}[Z|p]:=\int \! z \widehat{f}(z|p)dz$ (``FlexCode$_W$-SDM Mean''). The red and blue dots denote clusters with unimodal and multimodal mass densities, respectively. 		
						 {\em Center}: Boxplot of the the absolute fractional mass error $|\varepsilon|$  for the two populations. These results again indicate that much of the scatter in the mass error estimates are due to multimodal densities. {\em Right}: For comparison, we include a scatterplot of the predicted versus the true log mass masses for SDM regression as in \cite{ntampaka2015machine}, where we cannot extract this information.}
						\label{fig::clusterMultimodal}
					\end{figure}
					
	Finally, we notice that both the mean and the mode of FlexCode-SDM as well as FlexCode$_W$-SDM  densities improve upon plain SDM regression. Table \ref{tab::clusterPoint}  compares the fractional mass error distributions of the predictions.  By taking the mode of the FlexCode density we reduce the  $\varepsilon$ 68\% scatter\footnote{The  $\varepsilon$ 68\% scatter, $\Delta \varepsilon$, is the 68\% quantile of the distribution of $|\varepsilon|$} 
	 from $\Delta \varepsilon \approx 0.24$ for standard SDM down to a width of $\approx  0.15$ for FlexCode-SDM and of $\approx  0.17$ for FlexCode$_W$-SDM with a mode estimator.  

			  \begin{table}[H]
				\caption{Performance of different methods}
			  	\begin{center}
			  			  		\resizebox{\textwidth}{!}{%
			  		\begin{tabular}{l|*{3}{c}} &Mean fractional error &Median fractional error & 68\% scatter fractional error \\ \hline
			  		 \textit{SDM Regression} &0.052 &-0.004&0.244\\
			  		 \textit{FlexCode-SDM Mean}&0.012 &-0.036&-0.228\\
			  		 \textit{FlexCode-SDM Mode} &0.003&-0.025 &\textbf{0.152}\\
			  		 \textit{FlexCode$_W$-SDM Mean}&  -0.003 & -0.046 & 0.210 \\
			  		  \textit{FlexCode$_W$-SDM Mode} & $5.6*10^{-5}$ & -0.019  & \textbf{0.168}
			  		 
			  		\end{tabular}}
			  		\label{tab::clusterPoint}
			  	\end{center}
			  \end{table}

\comment{		
			\begin{figure}[H]
				\centering
				\subfloat[SDM Regression]{  \includegraphics[page=1,scale=0.3]{SDMCDEGaussSDMRegCodeMichelle.pdf}}
				\subfloat[FlexCode-SDM Mean]{ \includegraphics[page=1,scale=0.3]{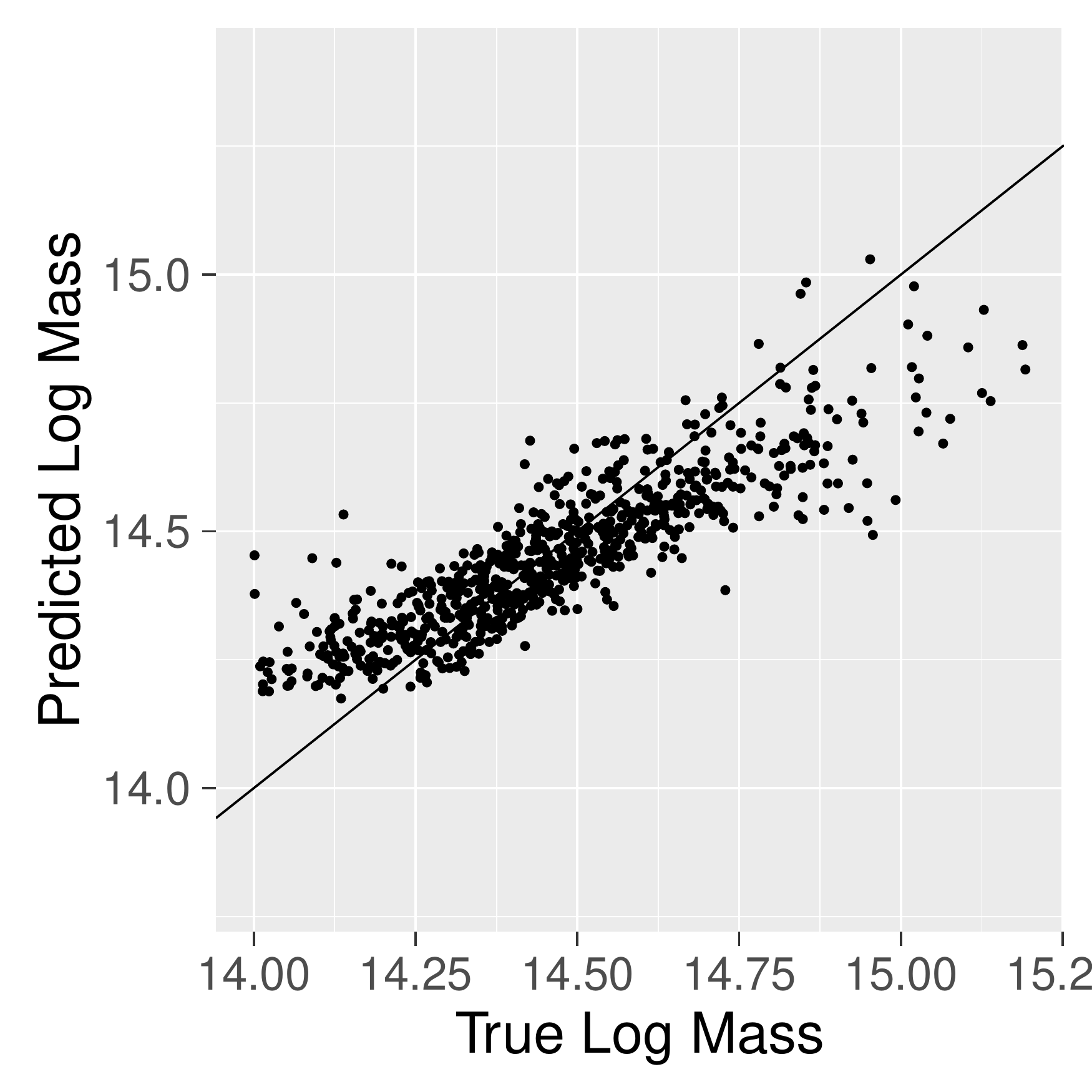}}
				\subfloat[FlexCode-SDM Mode]{ \includegraphics[page=10,scale=0.3]{SDMCDEGaussSDMFlexCodeMichelle.pdf}
				}
				\caption{\footnotesize  Left:  Predicted log masses versus true log mass masses when using {\bf (a)} SDM regression				as in \cite{ntampaka2015machine}), {\bf (b)} a mean estimator derived from FlexCode-SDM densities, and {\bf (c)} a mode estimator derived from the FlexCode densities. 
Both the FlexCode mean and mode estimators improve upon the SDM mass predictions. Moreover, with FlexCode one can better model the uncertainty in the mass predictions (Figure \ref{fig::clusterCoverage}, top left) and compute more informative HPD  regions (Figure \ref{fig::clusterCoverage}, top right).
}
				\label{fig::cluster2}
			\end{figure}
}

	 To summarize: 
	FlexCode extends SDM to conditional density estimation on distributions, and the estimated densities produce better point estimates of cluster masses. The real advantage with FlexCode, however, is that we can more accurately quantify the uncertainty in the predictions and potentially improve inference for outliers or cases that are not well described by one-number summaries. For example, 
	we can use the estimated densities to construct more informative highest predictive density (HPD) regions of the cluster mass, i.e., regions of the form
	$\{z:\widehat{f}(z|\x)\geq K\}$, where $K$ is chosen in such a way that the regions have the desired coverage level (e.g., 95\%). The top  panels of Figure~\ref{fig::clusterCoverage} shows some examples of multimodal densities and their 95\% HPD regions.  In many cases, returning {\em a predictive region for the cluster mass} is a better alternative to just taking the mean or mode of the density. 
	The coverage plot in the bottom right panel also indicates that the empirical coverage of these regions is indeed close to 95\%.

\section{Theory}
\label{sec::theory}


In this section, we derive bounds and rates for FlexCode; that is, the conditional density estimator in Eq.~\ref{eq:cdeEst}.
We use the notation $\widehat{f}_I(z|\x)$ to indicate its dependence on the cutoff $I$.  

We assume that 
$f$ belongs to a set of functions which are not too ``wiggly''. 
For every $ s>\frac{1}{2}$ and $0<c<\infty $, let $W_{\phi}(s,c)=\{f \!= \!\sum_{i\geq 1} \theta_i \phi_i \!:  \! \sum_{i\geq 1} a_i^2 \theta^2_i \leq c^2 \}$, where $a_i \! \sim \! (\pi i)^s$, denote the Sobolev space. For the Fourier basis
 $\{\phi_i\}_i$, this is the standard definition of Sobolev space \citep{wasserman}; it is  the
 space of functions that have their $s$-th weak derivative bounded by $c^2$ and integrable in $\mathcal{L}^2(\Re)$.
We enforce smoothness in the $z$-direction by requiring $f(z|\x)$ to be in a Sobolev space for all $\vec{x}$.
This is formally stated as Assumption~\ref{assump-sobolevZ}, where $\beta$
and $C$ are used to link the Sobolev spaces at different $x$.

\begin{Assumption} [Smoothness in $z$ direction]
 \label{assump-sobolevZ} $\forall \vec{x} \! \in \! \mathcal{X}$, 
$f(z|\x) \! \in \! W_{\phi}(s_\vec{x},c_\vec{x}),$ 
 where $f(z|\x)$ is viewed as a function of $z$, and $s_\vec{x}$ and $c_\vec{x}$ are such that 
$\inf_\vec{x} s_\vec{x}\overset{\mbox{\tiny{def}}}{=}\beta>\frac{1}{2}$ and 
$\int_\mathcal{X} c_\vec{x}^2d\vec{x}  \overset{\mbox{\tiny{def}}}{=} C <\infty$.
\end{Assumption}


We also assume that each  
   function $\beta_i(\x)$ is estimated using a regression method with convergence rate $O(n^{-2\alpha/(2\alpha+d)})$, where typically $\alpha$ is a parameter related to the smoothness of the $\beta_i(\x)$ function, and $d$ is either the  
     %
   number of relevant covariates or the intrinsic dimension of $\x$.
 In other words, we assume that each regression \emph{adapts}
to  sparse structure in the data. 
 This is formally stated as Assumption~\ref{assump-regression}.

\begin{Assumption} [Regression convergence]
 \label{assump-regression} For every $i \in \mathbb{N}$, there exists 
 some $d \in \mathbb{N}$ and $\alpha>0$ such that
 $$ \E\left[\int \left(\widehat{\beta}_i(\x)-\beta_i(\x) \right)^2d\x \right]=O(n^{-2\alpha/(2\alpha+d)})$$
\end{Assumption}

Note that the smoothness parameter $\alpha$ must be the same for every $i \in \mathbb{N}$.
 Typically this assumption will hold because in many applications
it is reasonable to assume that \
(i) if $\x_1$ is close to $\x_2$, then $f(z|\x_1)$ is also close to $f(z|\x_2)$ for every $z \in \Re$ (in other words,
$f(z|\x)$ is smooth as a function of $\x$),
 and (ii) there is some structure in x (e.g., low intrinsic dimensionality) or in the relationship between x and z (e.g., sparsity), which the regression method for estimating $\beta_i$ takes advantage of. Here are some examples where Assumption~\ref{assump-regression} holds:
\begin{enumerate}[label=(E{\arabic*})]
\item $\widehat{\beta}_i$ is the k-nearest neighbors estimator  \citep{kpotufe2011k}, $d$ is the intrinsic dimension of the covariate space
and,  for every $z \in [0,1]$,
   $f(z|\x)$ is $L$-Lipschitz in $\x$ (in this case, $\alpha=1$);
\item  $\widehat{\beta}_i$ is a local polynomial regression \citep{Bickel:Li:2007}, $d$ is the intrinsic dimension of the covariate space and, for every $z \in [0,1]$, $f(z|\x)$ is $\alpha$ times differentiable with all partial derivatives up to order $\alpha$ in $\x$ are bounded; 
\item  $\widehat{\beta}_i$ is the Rodeo estimator  \citep{lafferty2008rodeo}, $d$ is the number of variables that affect the distribution of $Z$ and, for every  $z$,  all partial derivatives of $f(z|\x)$ 
 up to fourth order 
  in $\x$ are bounded  (in this case, $\alpha=2$); 
\item  $\widehat{\beta}_i$ is the regression estimator from \citet{bertin2008selection}, $d$ is the number of variables that affect the distribution of $Z$,\footnote{That is, there exists a subset $R \subseteq \{1,
	\ldots,D\}$ with $|R|=d$ such that  $f(z|\x)=f(z|(x_i)_{i \in R})$} and, for every $z \in [0,1]$,  $f(z|\x)$
is $\alpha$-H\"olderian in $\x$; 
\item  $\widehat{\beta}_i$ is the Spectral series regression  \citep{LeeIzbickiReg}, $d$ is the intrinsic dimension of the covariate space and, 
 for every $z \in [0,1]$,  $f(z|\x)$ is smooth with respect to $P_X$ according to $\int ||\nabla f(z|\x)||^2dS(\x)<\infty$ for a smoothed version $S(\x)$ of $f$ (in which case $\alpha=1$);
\item  $\widehat{\beta}_i$ is a local linear functional regression \citep{baillo2009local}, the predictor $X$
is a function talking values in $\mathcal{L}^2([0,1])$, $X$
is fractal of order $\tau$, and, for every $z \in [0,1]$, $f(z|\x)$ is  
twice  differentiable with a continuous second derivative (yielding rates with $\alpha=2$ and $d=\tau$.)
\end{enumerate}

In essence, Assumption 2  holds for examples E1-E6 because smoothness in $f(z|\x)$ (seen as a function of $\x$) implies smoothness of the $\beta_i(\x)$ functions in FlexCode. We refer to Appendix A1  for details and proofs.
  (See also, e.g., \citet{yang2015minimax} and references therein for other adaptive regression methods.)
We also note that the converge rates may vary depending on the choice of basis.

Under Assumptions 1-2, we bound the bias and variance of $\widehat{f}_I(z|\x)$ separately.

\begin{Lemma} [Bias Bound] Under Assumption \ref{assump-sobolevZ},
$$ \sum_{i>I} \int\left(\beta_i(\x) \right)^2d\x=O(I^{-2\beta}) $$
 \label{biasBound}
\end{Lemma}


\begin{Lemma} [Variance Bound] From Assumption \ref{assump-regression}, it follows that
$$\sum_{i=1}^I \E\left[\int \left(\widehat{\beta}_i(\x)-\beta_i(\x) \right)^2d\x \right]=IO\left(n^{-2\alpha/(2\alpha+d)}\right)$$
 \label{varianceBound}
\end{Lemma}

Our main result follows.

\begin{thm} 
	\label{thm::main}
	Under Assumptions \ref{assump-sobolevZ} and \ref{assump-regression}, an upper bound on the risk of the CDE 
from Equation \ref{eq:cdeEst}  is  
 $$\E\left[\iint\left(\widehat{f}_I(z|\x)-f(z|\x)\right)^2dzd\x\right] \leq  IO\left(n^{-2\alpha/(2\alpha+d)}\right)+O(I^{-2\beta})$$
\end{thm}

See Appendix~\ref{sec::theoryAppendix} for proofs.


\begin{Cor} Under Assumptions \ref{assump-sobolevZ} and \ref{assump-regression}, it is optimal to take 
$$I \asymp n^{\frac{2\alpha}{(2\alpha+d)(2\beta+1)}},$$
which yields the rate 
$$O\left(n^{-\frac{2\beta}{2\beta+d\frac{2\beta+1}{2\alpha}+1}}\right)$$
for the estimator in Equation \ref{eq:cdeEst}.
\end{Cor}

To summarize: The convergence rate of FlexCode only depends on $d$, the ``true'' dimension of the problem.
Moreover, the rate is near minimax with regards to $d$:
In the isotropic setting where $\x$ and $z$ have the same degree of smoothness, i.e., $\alpha=\beta$, the rate becomes
$$O\left(n^{-\frac{2\alpha}{2\alpha+d\frac{2\alpha+1}{2\alpha}+1}}\right),$$
which is close
to the minimax rate $O\left(n^{-\frac{2\alpha}{(2\alpha+1+d)}}\right)$ 
of a conditional density estimator with $d$ covariates \citep{IzbickiLeeCDE}. The difference is the multiplicative factor $\frac{2\alpha+1}{2\alpha}$, which gets closer to 1, the smoother $f$ is. 
Although FlexCode's rate is 
slightly slower than the optimal rate,\footnote{The reason may be that that we optimize the tuning parameters of each regression $\widehat{\beta}_i(\x)$ so as to have optimal regression estimates (Assumption \ref{assump-regression}) rather than an optimal estimate of $f(z|\x)$.} the estimator is still considerably faster than
$O\left(n^{-\frac{2\alpha}{(2\alpha+1+D)}}\right)$, the usual minimax rate of a nonparametric conditional density estimator in $\mathbb{R}^D$.
In other words, even though there are $D$ covariates, our estimator can overcome the curse-of-dimensionality 
and behave as if there are only $d \ll D$
covariates.


Finally, note that although we here restrict our examples to cases where
either (i) the intrinsic dimension is small or (ii) several
covariates are irrelevant, the theory we
develop can easily be applied to other
settings for high-dimensional regression estimation. For instance, 
\citet{yang2015minimax} introduce a third type of sparse structure: in their paper, 
$r$ may depend on all $D$ covariates, but admits an additive structure 
$r = \sum_{s=1}^k r_s,$ where each component function
$r_s$ depends on a small number $d_s$  of predictors.
The authors then show that an additive Gaussian process regression
achieves good rates of convergence in such a setting. It follows that
FlexCode can achieve good rates under a similar additive setting 
$f(z|\x) = \sum_{s=1}^k f_s(z|\x)$ 
if one estimates the expansion coefficients via additive Gaussian process regression.

\section{Conclusions}
\label{sec::concl}
 With FlexCode, one can use {\em any} regression methodology to estimate a conditional density. In other words, FlexCode is a powerful inference and data analysis tool that converts prediction to the problem of understanding the role of covariates in {\em explaining} the outcome, with meaningful measures of uncertainty attached to the predictions. Because of the flexibility of the method, one can construct estimators for a range of different scenarios with complex, high-dimensional data. In the paper, we emphasized examples where several redundant covariates are correlated, and examples where only a small number of covariates
influence the distribution of the response. We showed that
FlexCode has good theoretical properties and empirical performance comparable to 
state-of-the-art approaches in a wide variety of settings, including cases with mixed data types and functional data.

In the paper, we restricted most analyses to Fourier bases in the outcome space, but for distributions that are inhomogeneous with respect to the response variable, one may benefit from nonlinear approximations in a wavelet basis \citep{mallat1999wavelet}. We will explore this aspect further in a separate paper, as well as extensions of FlexCode to approximate likelihood computation  for structured data and complex simulation models. Another interesting direction for future work is variable selection via FlexCode. 
For example, FlexCode-Forest and FlexCode-SAM currently perform a separate variable selection for each coefficient $\beta_i(\x)$
in FlexCode (Eq. \ref{eq::coeff}), but one can unify these results to define a common support for the final FlexCode estimate.
\vspace{2mm}

\comment{In this work we propose a promising new conditional density estimator
which takes advantage of the large literature on regression function estimation, FlexCode.
The estimator is able to overcome the curse-of-dimensionality
in several scenarios. We emphasized two example:
(i) when the covariates are very redundant, and (ii)
when only a small number of covariates
influence the distribution of the response.
We showed that
FlexCode has good theoretical properties, and empirical performance comparable to 
state-of-the-art approaches in a wide variety of settings,
including mixed data types and functional data examples. Moreover, it is generally faster than 
many of such approaches. Finally, we argued
that it is extremely flexible. 

FlexCode is
a very powerful tool: it
extends \emph{any} regression estimator to a conditional density estimator, i.e., given any regression
methodology, our framework allows one to use the same methodology to estimate a conditional density. 
}



{\small \noindent \textbf{Acknowledgments.}
	We thank Michelle Ntampaka and Hy Trac for sharing the data for the galaxy cluster mass example, and Peter E. Freeman for his help with the photometric redshift and galaxy cluster mass studies. This work was partially supported by 
	\emph{Funda\c{c}\~ao de Amparo \`a Pesquisa do Estado de S\~ao Paulo} (2014/25302-2), NSF DMS-1520786, and the National Institute of Mental Health grant R37MH057881.}\\

\bibliographystyle{plainnat}

{ \footnotesize 
	\bibliography{paper_arxiv}
}

\newpage

\appendix

\textbf{\huge Appendix}

\section{Diagnostic Test of  Conditional Density Estimates}\label{sec::Appendix_diagnostics}


To assess how well a model actually fits the observed data, we use coverage plots that are based on {\em Highest-Predictive Density (HPD) regions}.
\comment{
\begin{enumerate}
 \item(\textbf{Q-Q Plot}) For every $c$ in a grid of values on $[0,1]$ and for every observation $i$ in the test sample, compute $Q_i^c = \widehat{F}_{z|\vec{x}_i}^{-1}(c)$. Define
 $\widehat{c}=\frac{1}{n}\sum_{i=1}^n\I(Z_i \leq Q_i^c).$ We plot the values of $\widehat{c}$ against the corresponding values of $c$. If the distributions $\widehat{F}_{z|\vec{x}}$ and $F_{z|\vec{x}}$ are similar, then the points in the Q-Q plot will approximately lie on the line $\widehat{c} = c$.

 \item(\textbf{P-value})  For every test data point $i$, let 
 $U_i=\widehat{F}_{z|\vec{x}_i}(Z_i).$
 If the data are really distributed according to $\widehat{F}_{z|\vec{x}}$, then 
 $U_1,\ldots,U_n \overset{\mbox{\tiny{iid}}}{\sim} Unif(0,1)$. Hence, we compute the p-value for a Kolmogorov-Smirnoff test that compares the distributions of these statistics to the uniform distribution. 
 
 \item \textbf{Coverage Plots and HPD Regions.}
 }
 
 Let $\widehat{f}_{z|\vec{x}_i}$ denote the estimated conditional density function for $z$ given $\vec{x}_i$. For every $\alpha$ in a grid of values in  $[0,1]$  and for every data point $i$ in the test sample, we define a set $A_i$ such that 
 $\int_{A_i}\widehat{f}(z|\vec{x}_i)dz=\alpha.$ Here we choose the set $A_i$ with the smallest area: $A_i=\{z:f(z|\x_i)>t\}$ 
 where $t$ is such that  $\int_{A_i}\widehat{f}(z|\vec{x}_i)dz=\alpha$; i.e., $A_i$
 is a Highest Predictive Density region.

 Let 
 $\widehat{\alpha}_i=\frac{1}{n}\sum_{i=1}^n \I(Z_i \in A_i).$ 
 If $\widehat{f}_{z|\vec{x}}$ and the true density $f_{z|\vec{x}}$ are similar, then $\widehat{\alpha}_i \approx \alpha_i$.
 Hence, as a diagnostic tool, we graph $\widehat{\alpha}_i$ versus $\alpha_i$ for the test set, and assess how close these points are to the line $\widehat{\alpha} = \alpha$. For each $\alpha_i$, we also include a $95\%$ confidence interval
 based on a normal approximation to the binomial distribution.

\section{Additional Twitter Data}

	Here we consider 5000 geotagged tweets posted in July 2015 that 
	 include either the keyword  \emph{frio} or the keyword \emph{calor}; these words mean
	cold and hot in Spanish as well as in Portuguese. As in Sec.~\ref{sec::twitter}, the goal is to predict the latitude
	and longitude of a tweet, $\z$, based on its content $\x$.
	Using the same methodology (\emph{FlexCode-SAM}) as before, we estimate $f(\z|\x)$. 
	
	 Fig.~\ref{fig::twitterCalor} shows the results for three tweets. In the tweet corresponding to the left plot, the user mentions ``beach" and ``heat".
	Because (i) July is a summer month in the north hemisphere, (ii) the tweet is in Spanish, and (iii) it mentions ``beach",
	FlexCode automatically assigns high probability to the coast of Spain. For the example corresponding to the middle plot, on the other hand,
	the word ``beach" does not occur, but the tweet is in Spanish and it mentions hot weather. As a result,
	our density model assigns high probability to the interior of Spain. Our final example, corresponding to the right plot in the figure,
	is a tweet in Portuguese about cold weather. Our FlexCode model here assigns high probability to big cities in Brazil, which is consistent with July being a winter month in the south hemisphere.   
	We also notice that it in the winter rains a lot in Recife, the northernmost city that are colored red in the density plot. This is why FlexCode assigns a high probability to this location despite the city being much smaller than Sao Paulo and Rio de Janeiro.

	\begin{figure}[H]
	  \centering \hspace{-14mm}
	   \includegraphics[page=6,scale=0.35]{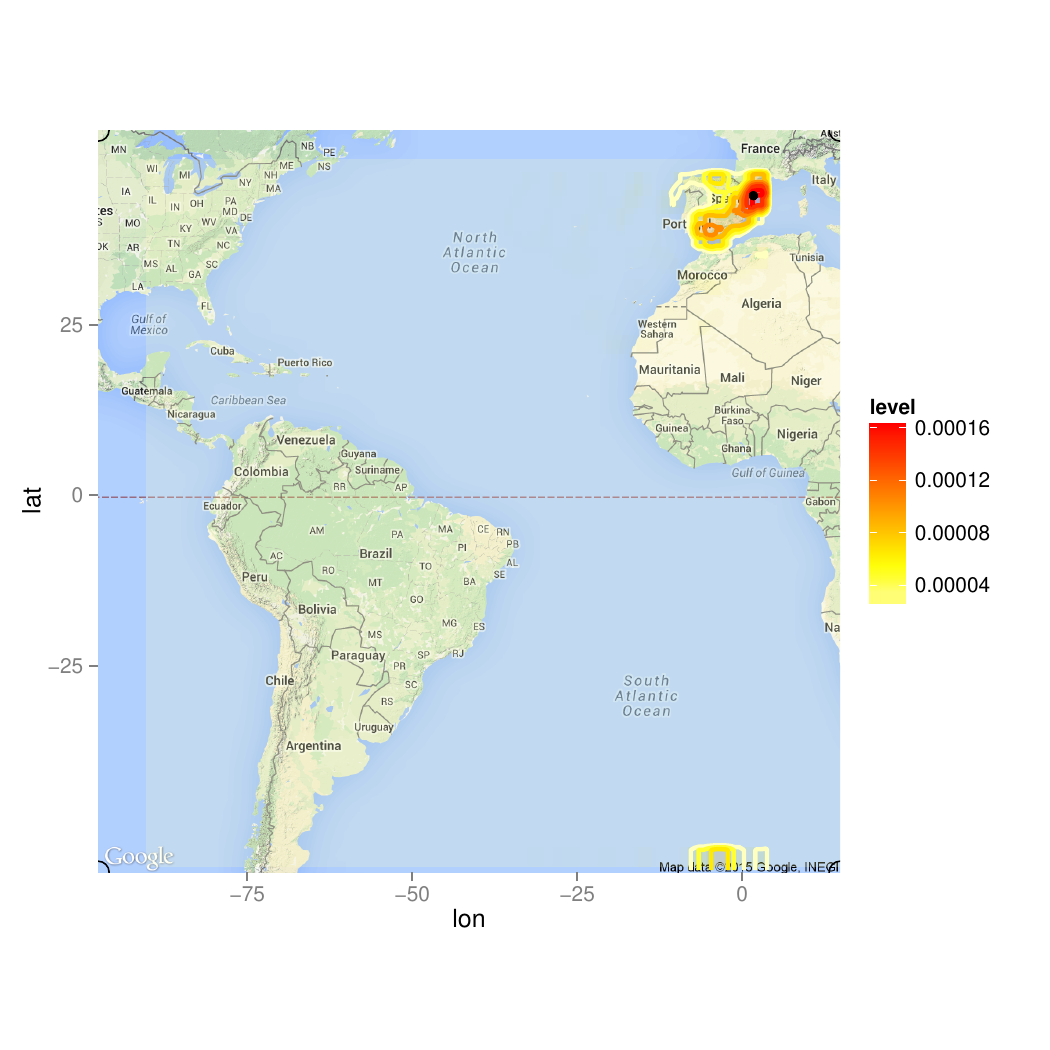}\hspace{-6mm}
	   	   \includegraphics[page=11,scale=0.35]{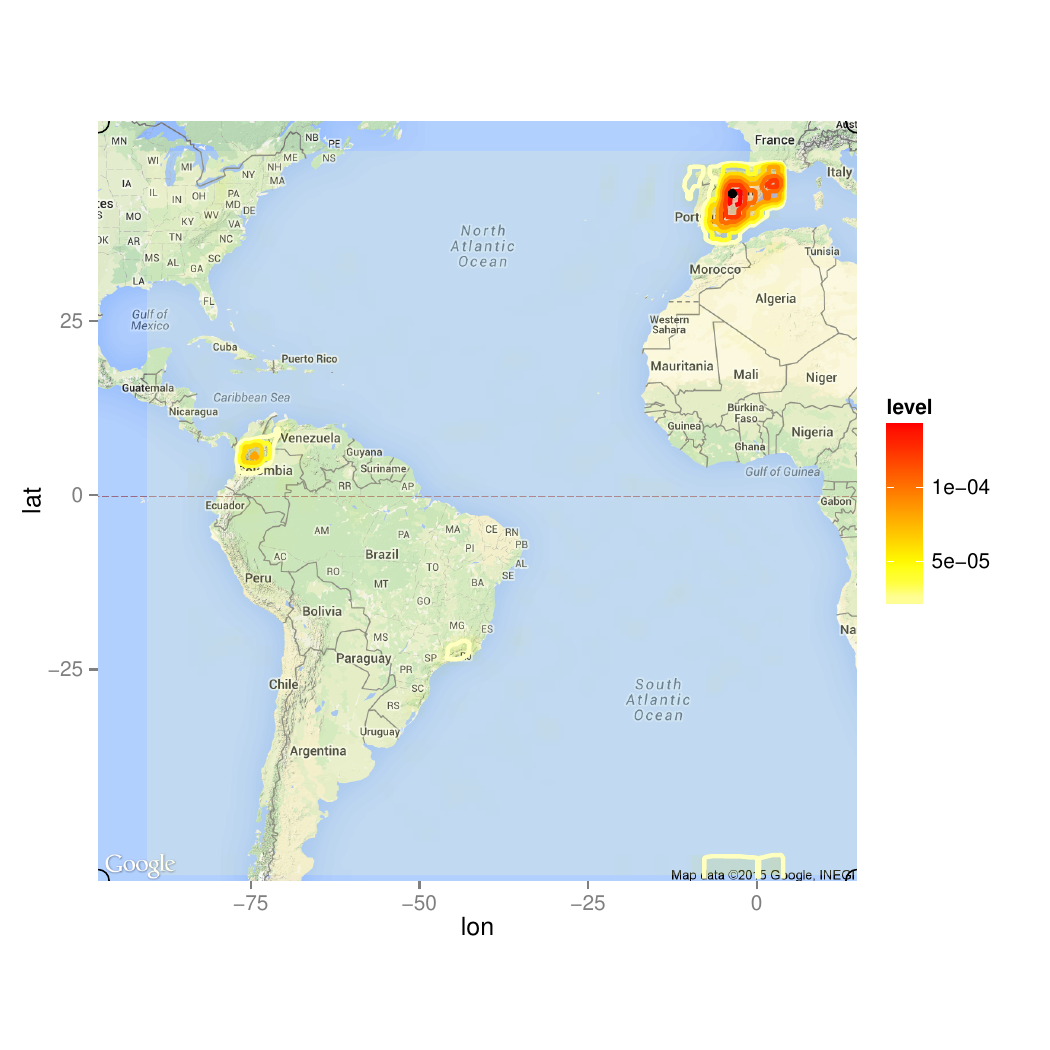}\hspace{-6mm}
	   \includegraphics[page=21,scale=0.35]{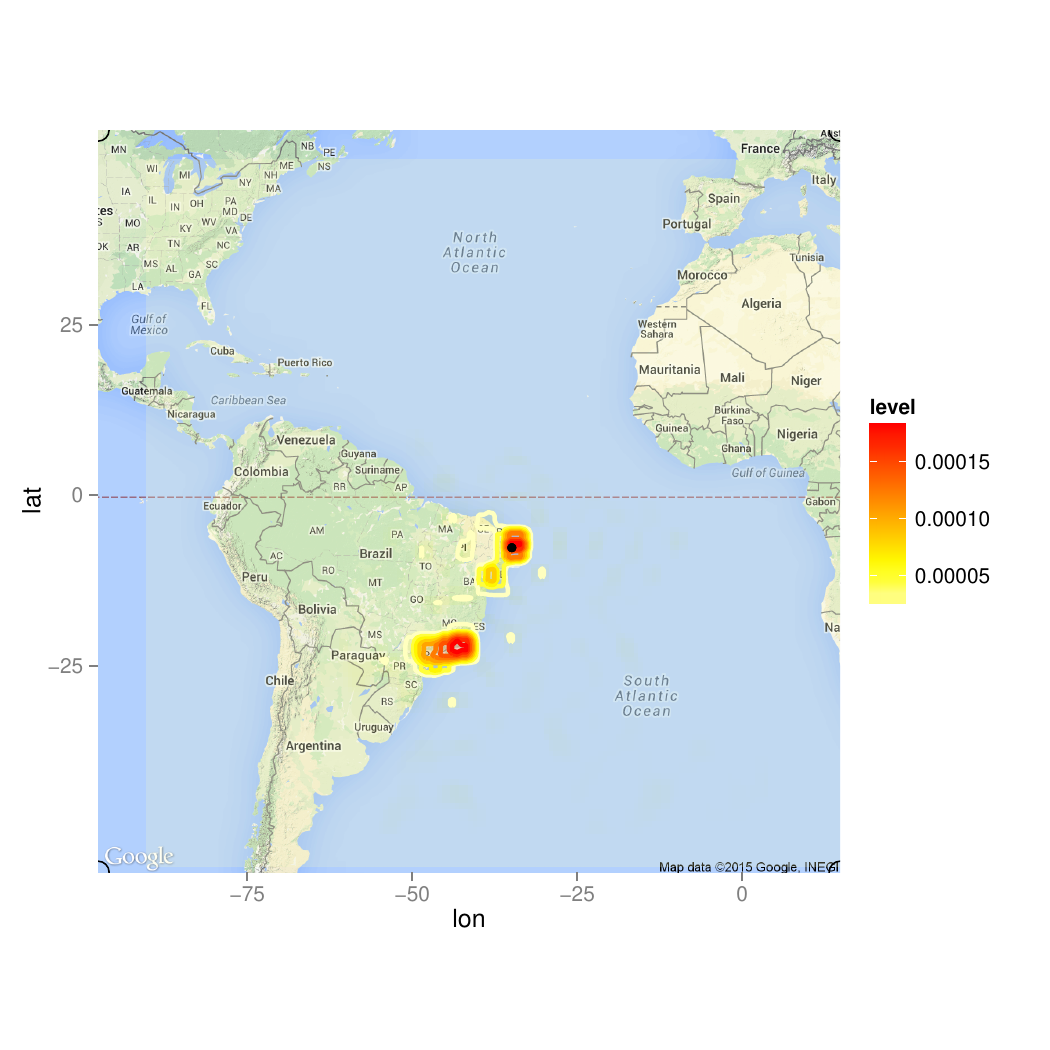}	   \hspace{-6mm}
	   \caption{\footnotesize Level sets of the estimated probability densities of the location of three tweets given their contents.
	   The black dots show the true location where each tweet was posted.
	   Left: \emph{Contra la ola de calor\ldots Un chapuzón en la playa tempranito y ahora\ldots Reclusión en casa\ldots} (Against the heat wave\ldots 
	   A dip in the beach very early and now\ldots  Confinement at home). Middle:  \emph{Combatiendo el calor \#verano \#lacuevadekrusty \#elmolar} (Fighting the heat \#summer \#lacuevadekrusty \#elmolar). Right: 
	   \emph{Domingo de chuva, frio gostoso é dia de: Fazer planilha do DVD kkkkk} (Rainy Sunday, pleasant cold is a day of: Making a DVD playlist lol).	}
	   \label{fig::twitterCalor}
	\end{figure}

	\section{Proofs and Additional Results}
	\label{sec::theoryAppendix}

	To prove that the estimators in examples E1-E6 in Sec.~\ref{sec::theory} satisfy Assumption~\ref{assump-regression}, we only need to show that smoothness in the conditional density  $f(z|\x)$ (seen as a function of $\x$) implies smoothness for each varying coefficient  $\beta_i(\x)$. 
Assumption~\ref{assump-regression} then follows directly from known convergence results for regression. For E3 and E4, note that if there exists a subset $R \subseteq \{1,
	\ldots,D\}$ with $|R|=d$ such that  $f(z|\x)=f(z|(x_i)_{i \in R})$
	(i.e., there are only $d$ relevant covariates), then $\beta_i(\x)=\beta_i((x_i)_{i \in R})$.

	Different estimators
	use different notions of smoothness. In~\citet{kpotufe2011k}, the authors show that k-NN regressors converge at rates the depend only on the intrinsic dimension of data if the target function is Lipschitz. 
	Hence, for example E1, we use the Lipschitz notion of smoothness:

	\begin{Lemma} Let $\{\phi_i\}_i$ be the Fourier basis. If, for every fixed $z \! \in \! \Re$, 
		$f(z|\x)$ is $L$-Lipschitz function, then $\beta_i(\x)$ is $\sqrt{2}L$-Lipschitz for all $i \in \mathbb{N}$.
		
		\label{lemma:lip}
	\end{Lemma}
	
	\begin{proof}
		Let $\x,\y \in \Re^D$. Then  
		\begin{align*}
		|\beta_i(\x)-\beta_i(\y)| &= \left|\int \! \phi_i(z)f(z|\x)dz-\int \! \phi_i(z)f(z|\y)dz \right|  \leq \int \! \left| \phi_i(z)\right| \left| f(z|\x)-f(z|\y)  \right|dz  \\
		&\leq  L ||\x-\y||  \int \! \left| \phi_i(z)\right|  dz \leq  \sqrt{2} L ||\x-\y||  \int \! \left| \phi_i(z)\right|^2 dz  \\
		& = \sqrt{2} L ||\x-\y||
		\end{align*}
		

	\end{proof}

	Local polynomial regression \citep{Bickel:Li:2007} and Rodeo \citep{lafferty2008rodeo} use the notion
	of bounded partial derivatives. Hence, we use the following result:

	\begin{Lemma} Let $\{\phi_i\}_i$ be the Fourier basis. If for every fixed $z \! \in \! \Re$, $f(z|\x)$ has 
		all partial derivatives of order $p$ bounded by $K$, then $\beta_i(\x)$ has 
		all partial derivatives of order $p$ bounded by $\sqrt{2}K$
		\label{lemma:boundedDeriv}

	\end{Lemma}
	
	\begin{proof}
		Let $\x \! \in \! \Re^D$ and $a_1,\ldots,a_p \in \{1,2,\ldots,D\}$. Then
		\begin{align*}
		\left| \frac{\partial}{\partial x_{a_1}\ldots \partial x_{a_p}} \beta_i(\x) \right| &= \left|\frac{\partial}{\partial x_{a_1}\ldots \partial x_{a_p}}  \int\phi_i(z)f(z|\x)dz \right|  \leq \int\left| \phi_i(z)\right| \left| \frac{\partial}{\partial x_{a_1}\ldots \partial x_{a_p}} f(z|\x)   \right|dz  \\
		&\leq \sqrt{2} K 
		\end{align*}
	\end{proof}

	The notion of smoothness in \citet{bertin2008selection} is based on H\"olderian classes. Hence:

	\begin{Lemma} Let $\{\phi_i\}_i$ be the Fourier basis
		and $\mathcal{P}_l(f)(\cdot ,\x)$ be Taylor polynomial of order $l$ associated with $f$ at
		the point $\x$.
		If, for every fixed $z \! \in \! \Re$, 
		$f_z(\x):=f(z|\x)$ belongs to $\Sigma(\alpha,L)$, the $\alpha$-H\"olderian class, i.e., 
		$|f_z(\x)-\mathcal{P}_l(f_z)(\t ,\x)|\leq L ||\t-\x||_1^\alpha$ where $l =  \lfloor  \alpha \rfloor$,
		then $\beta_i(\x)$ belongs to $\Sigma(\alpha,\sqrt{2}L)$ for all $i \in \mathbb{N}$.
		
		\label{lemma:holder}
	\end{Lemma}

	\begin{proof}
		Because $\beta_i(\x)=\int\phi_i(z)f(z|\x)dz,$ then $\mathcal{P}_l(\beta_i)(\t ,\x)=\int\phi_i(z)\mathcal{P}_l(f_z)(\t,\x) dz $.  Hence, we have that
		\begin{align*}
		|\beta_i(\x)-\mathcal{P}_l(\beta_i)(\t ,\x)| \leq \int| \phi_i(z)| \ |f(z|\x)-\mathcal{P}_l(f_z)(\t,\x)| dz \leq 
		\sqrt{2} L ||\t-\x||_1^\alpha
		\end{align*}
	\end{proof}

	The spectral series estimator \citep{LeeIzbickiReg}
	assumes that the regression function
	is smooth with respect to $P$. Hence:
	
	\begin{Lemma} Let $\{\phi_i\}_i$ be the Fourier basis
		and assume  that,
		for every fixed $z \! \in \! \Re$, $\int ||\nabla f(z|\x)||^2dS(\x)  <\infty$.
		Then, for all $i \in \mathbb{N}$, 
		$\int ||\nabla \beta_i(\x)||^2dS(\x) < \infty$.
		\label{lemma:wrtP}
	\end{Lemma}

	\begin{proof}
		Because $\beta_i(\x)=\int\phi_i(z)f(z|\x)dz,$ then  
		\begin{align*}
		\int ||\nabla \beta_i(\x)||^2dS(\x)  &=
		\int \left\|\nabla \int \! \phi_i(z)f(z|\x)dz\right\|^2dS(\x) =
		\int \left\|\int\! \phi_i(z)\nabla f(z|\x)dz\right\|^2dS(\x)  \\
		&\leq \int \left( \int \phi_i^2(z) dz \right) \left(\int\! || \nabla f(z|\x)||^2dz \right) dS(\x) \\
		&= \int \left(\int\! || \nabla f(z|\x)||^2 dS(\x) 
		\right)dz < \infty
		\end{align*}
	\end{proof}
	
	Finally, the local linear functional regression estimator \citep{baillo2009local}
	assumes that the regression function
	has continuous second derivatives. Hence:

	\begin{Lemma} Let $\{\phi_i\}_i$ be the Fourier basis
		and assume  that $\x \in \mathcal{L}^2([0,1])$ and that,
		for every fixed $z \! \in \! \Re$, $f(z|\x)$ has continuous second derivative. Then 
		$\beta_i(\x)$  also has continuous second derivative for every $i\in \mathbb{N}$.
		\label{lemma:continuousSecond}
	\end{Lemma}

	\begin{proof}
		Because $\beta_i(\x)=\int\phi_i(z)f(z|\x)dz,$ then  
		\begin{align*}
		\frac{d^2 \beta_i(\x)}{d\x^2} = 	\int\phi_i(z) \frac{d^2 f(z|\x)}{d\x^2} dz
		\end{align*}
	\end{proof}
	
	We now present the proofs of the other results presented
	in the paper.
	
	\subsection{Proof of Lemma 1}

	\begin{proof}
		Because $f(z|\x)$ belongs to $W_{\phi}(s_\vec{x},c_\vec{x})$ for all $z$, 
		and $f(z|\x)=\sum_{i \geq 1} \beta_i(\x)\phi_i(z)$,
		we have that
		\begin{align*}
		\sum_{i \geq I} I^{2s_\vec{x}}\left(\beta_i(\x)\right)^2 \leq \sum_{i \geq I} i^{2s_\vec{x}}\left(\beta_i(\x)\right)^2 \leq  c_\vec{x}^2.
		\end{align*}
		Hence
		\begin{align*}
		\sum_{i \geq I} \int \left(\beta_i(\x)\right)^2 d\vec{x}\leq \int \frac{ c_\vec{x}^2}{I^{2s_\vec{x}}}d\vec{x}
		= O(I^{-2 \beta}).
		\end{align*}
		
	\end{proof}

	\subsection{Proof of Theorem 1:}

	\begin{proof}
		\begin{align*}
		&\iint\left(\widehat{f}_I(z|\x)-f(z|\x)\right)^2dzd\x=\\  
		&\iint\left(\sum_{i=1}^I\widehat{\beta}_{i}(\x)\phi_i(z)-\sum_{i\geq 1} \beta_{i}(\x)\phi_i(z)\right)^2dzd\x=\\
		&\iint\left(\sum_{i=1}^I( \widehat{\beta}_{i}(\x) -\beta_i(\x))\phi_i(z)-\sum_{i>I} \beta_{i}(\x)\phi_i(z)\right)^2dzd\x \stackrel{(*)}{=}\\
		&\int \left( \sum_{i=1}^I ( \widehat{\beta}_{i}(\x) -\beta_i(\x))^2 +  \sum_{i>I} (\beta_{i}(\x))^2 \right) d\x =\\
		& \sum_{i=1}^I \int ( \widehat{\beta}_{i}(\x) -\beta_i(\x))^2 d\x+  \sum_{i>I} \int (\beta_{i}(\x))^2  d\x, 
		\end{align*}
		where step $(*)$ follows from expanding the square and the fact that the Fourier basis is orthonormal 
		(i.e., the cross products in the expansion are zero).
		
		The final result follows from Lemmas~\ref{biasBound} and~\ref{varianceBound}.
	\end{proof}

\end{document}